\documentclass[11pt]{llncs}

\usepackage{amssymb,amsmath,stmaryrd}
\usepackage{gastex}
\usepackage{paralist,fullpage}
\usepackage{comment}
\usepackage{algorithmic}
\usepackage[ruled, vlined]{algorithm2e}
\usepackage{color}

\dontprintsemicolon

\newcommand{\nat}{\mathbb{N}}
\newcommand{\integ}{\mathbb{Z}}
\newcommand{\rat}{\mathbb{Q}}

\newcommand{\cur}{\mathit{cur}}

\newcommand{\attr}{\mathsf{Attr}}
\newcommand{\AttrOne}[1]{\attr_1(#1)}
\newcommand{\AttrTwo}[1]{\attr_2(#1)}

\newcommand{\Attri}[1]{\attr_i(#1)}
\newcommand{\set}[1]{\{#1\}}

\newcommand{\win}[1]{\langle \! \langle #1 \rangle\! \rangle}

\newcommand{\integers}{\mathbb{Z}}

\newcommand{\MeanPayoff}{{\sf MeanPayoff}}
\newcommand{\EL}{{\sf EL}}

\newcommand{\Plays}{{\sf Plays}}

\newcommand{\Nat}{\mathbb{N}}

\newcommand{\MeanPayoffInfSup}{{\sf MeanPayoffInfSup}}
\newcommand{\MeanPayoffSup}{{\sf MeanPayoffSup}}
\newcommand{\MPsup}{\overline{{\sf MP}}}
\newcommand{\MeanPayoffInf}{{\sf MeanPayoffInf}}
\newcommand{\MPinf}{\underline{{\sf MP}}}

\newcommand{\Last}{{\sf Last}}
\newcommand{\outcome}{\mathsf{outcome}}

\def\abs#1{\ensuremath{\lvert #1\rvert}}
\newcommand{\tuple}[1]{\langle #1 \rangle}

\newcommand{\PosEnergy}{{\sf PosEnergy}}

\renewcommand{\qed}{\hfill \ensuremath{\Box}}
\def\red#1{\textcolor{red}{#1}}

\def\mynote#1{{\sf $\clubsuit$ #1 $\clubsuit$}}
\def\mynote#1{{\sf \red{$\clubsuit$} #1 \red{$\clubsuit$}}}

\newcommand{\SetCycle}{\mathit{SetCycle}}
\newcommand{\wh}{\widehat}

\long\def\symbolfootnote[#1]#2{\begingroup%
\def\thefootnote{\fnsymbol{footnote}}\footnote[#1]{#2}\endgroup}

\title{The Complexity of Multi-Mean-Payoff and Multi-Energy Games\thanks{Preliminary 
versions appeared in the \emph{Proceedings of the IARCS Annual Conference on Foundations of Software Technology and Theoretical Computer Science} (FSTTCS), Schloss Dagstuhl - Leibniz-Zentrum fuer Informatik, LIPIcs, 2010, pp. 505-516, and in the \emph{Proceedings of the 14th International Conference on Foundations of Software Science and Computational Structures} (FoSSaCS), Lecture Notes in Computer Science 6604, Springer, 2011, pp. 275-289.}${}^{,}$\thanks{Corresponding author: Laurent Doyen;
address: LSV - ENS Cachan, 61 av. du President Wilson, 94235 Cachan Cedex, France; email: doyen@lsv.ens-cachan.fr.}
}
\author{
Yaron Velner\inst{1} \and
Krishnendu Chatterjee\inst{2} \and Laurent Doyen\inst{3} \and Thomas A. Henzinger\inst{2} 
\and Alexander Rabinovich\inst{1}
\and Jean-Fran\c{c}ois Raskin\inst{4}
}

\institute{
The Blavatnik School of Computer Science, Tel Aviv University, Israel \\
\and
IST Austria (Institute of Science and Technology Austria) \\
\and LSV, ENS Cachan \& CNRS, France \\
\and D\'epartement d'Informatique, Universit\'e Libre de Bruxelles (U.L.B.)
}

\begin{document}

\maketitle

\begin{abstract}
In mean-payoff games, the objective of the protagonist is to ensure that the 
limit average of an infinite sequence of numeric weights is nonnegative. 
In energy games, the objective is to ensure that the running sum of weights 
is always nonnegative. Multi-mean-payoff and multi-energy games replace 
individual weights by tuples, and the limit average (resp.\ running sum) 
of each coordinate must be (resp.\ remain) nonnegative. These games have 
applications in the synthesis of resource-bounded processes with multiple 
resources.

We prove the finite-memory determinacy of multi-energy games and show 
the inter-reducibility of multi-mean-payoff and multi-energy games for 
finite-memory strategies. We also improve the computational complexity 
for solving both classes of games with finite-memory strategies: while 
the previously best known upper bound was EXPSPACE, and no lower 
bound was known, we give an optimal coNP-complete bound. 
For memoryless strategies, we show that the problem of deciding 
the existence of a winning strategy for the protagonist is NP-complete.
Finally we present the first solution of multi-mean-payoff games with 
infinite-memory strategies.
We show that multi-mean-payoff games with mean-payoff-sup objectives can 
be decided in NP $\cap$ coNP, whereas multi-mean-payoff games with
mean-payoff-inf objectives are coNP-complete. 
\end{abstract}


%%\mynote{L: make uniform the numbering of Lemma and Theorems.}

\section{Introduction}

\noindent{\em Graph games and multi-objectives.} 
Two-player games on graphs are central in many applications
of computer science.  For example, in the synthesis problem, implementations of
reactive systems are obtained from winning strategies in games with a qualitative 
objective formalized by an $\omega$-regular specification~\cite{RW87,PnueliR89,AbadiLW89}.
In these applications, the games have a qualitative (boolean) objective 
that determines which player wins. 
On the other hand, games with quantitative objective which are natural models 
in economics (where players have to optimize a real-valued payoff) 
have also been studied in the context of automated design~\cite{Sha53,Condon92,ZP96}.
In the recent past, there has been considerable interest in the design of 
reactive systems that work in resource-constrained environments 
(such as embedded systems). 
The specifications for such reactive systems are quantitative, and 
give rise to quantitative games. 
In most system design problems, there is no unique objective to be optimized, 
but multiple, potentially conflicting objectives.
For example, in designing a computer system, one is interested not
only in minimizing the average response time but also the average power 
consumption.
In this work we study such multi-objective generalizations of the 
two most widely used quantitative objectives in games, namely,
\emph{mean-payoff} and 
\emph{energy} objectives~\cite{EM79,ZP96,CAHS03,BFLMS08}.

\smallskip\noindent{\em Multi-mean-payoff games.} 
A {\em multi-mean-payoff game} is played on a finite weighted game graph by two players. 
The vertices of the game graph are partitioned into positions that belong to 
player~$1$ and positions that belong to player~$2$. 
Edges of the graphs are labeled with $k$-dimensional vectors $w$ of integer values, i.e., 
$w \in \mathbb{Z}^k$. The game is played as follows. A pebble is placed on a 
designated initial vertex of the game graph. The game is played in rounds in which
the player owning the position where the pebble lies
moves the pebble to an adjacent position of the graph using an outgoing edge. 
%From there, a new round starts. 
The game is played for an infinite number of rounds, resulting in an infinite 
path through the graph, called a play. The value associated to a play is the 
mean value in each dimension of the vectors of weights labeling the edges of the play.
Accordingly, the winning condition for player~1
is defined by a vector of rational values $v \in \mathbb{Q}^k$ that specifies a 
threshold for each dimension. A play is winning for player~$1$ if its vector 
of mean values is at least~$v$. All other plays are winning for player~$2$,
thus the game is zero-sum. We are interested in the problem of deciding
the existence of a winning strategy for player~$1$ in multi-mean-payoff games. 
In general infinite memory may be required to win
multi-mean-payoff games, but in many practical applications such as the 
synthesis of reactive systems with multiple resource constraints, the 
multi-mean-payoff games with finite memory is the relevant problem.
Also they provide the framework for the synthesis of specifications 
defined by mean-payoff conditions~\cite{AlurDMW09,Concur10}, and the synthesis question for such 
specifications under \emph{regular (ultimately periodic)} 
words correspond to multi-mean-payoff games with finite-memory strategies.
Hence we study multi-mean-payoff games both for general strategies as well as 
finite-memory strategies.

\smallskip\noindent{\em Multi-energy games.} 
% {\em Generalized energy games} are played on the same game structures as generalized 
% multi-payoff games % not defined
In multi-energy games, the winning condition for player~1 requires that,
given an initial credit $v_0 \in \nat^k$, the sum of $v_0$ and all the 
vectors labeling edges up to position $i$ in the play is nonnegative, for all 
$i \in \nat$. The decision problem for multi-energy games
asks whether there exists an initial credit $v_0$ and a strategy for player~1 
to maintain the energy nonnegative in all dimensions against all strategies 
of player~2.

\smallskip\noindent{\em Contributions.}
In this paper, we study the strategy complexity and computational complexity 
of solving multi-mean-payoff and multi-energy games. 
The contributions are as follows.

First, we show that multi-energy and multi-mean-payoff games 
are determined when played with finite-memory strategies.
When considering finite-memory strategies, those games correspond 
to the synthesis question with ultimately periodic words, and they enjoy pleasant 
mathematical properties like existence of the limit of the mean value of the weights. 
We also establish that multi-energy and multi-mean-payoff games are not 
determined for memoryless strategies. 
Additionally, we show for multi-energy games determinacy under finite-memory coincides 
with determinacy under arbitrary strategies, and each player has a
winning strategy if and only if he has a finite-memory winning strategy.
In contrast, we show for multi-mean-payoff games that determinacy under finite-memory 
and determinacy under arbitrary strategies do not coincide. 
Moreover, for multi-mean-payoff games when the strategies for player~1 is 
restricted to finite-memory strategies, the winning set for player~1 
remains unchanged irrespective of whether we consider finite-memory 
or infinite-memory counter strategies for player~2.

%, and hence we focus on the study of multi-mean-payoff and energy games with finite-memory strategies. 

Second, we show that under the hypothesis that both players play either 
finite-memory or both play memoryless strategies, the decision problems 
for multi-mean-payoff games and multi-energy games are equivalent.

Third, we study the computational complexity of the decision problems 
for multi-mean-payoff games and multi-energy games, both for finite-memory strategies 
and the special case of memoryless strategies. 
Our complexity results can be summarized as follows.
(A)~For finite-memory strategies, we provide a nondeterministic polynomial-time 
algorithm %that makes a call to an oracle that is nondeterministic polynomial time algorithm 
for deciding negative instances of the problems\footnote{Negative  
instances are those where player~1 is losing, and by determinacy under finite-memory where player~2 is winning.}.
Thus we show that the decision problems are in coNP. 
This significantly improves the complexity as compared to the EXPSPACE 
algorithm that can be obtained by reduction to {\sc Vass} (vector addition systems with states)~\cite{BJK10}.
Furthermore, we establish a coNP lower bound for these problems by reduction 
from the complement of the  3SAT problem, hence showing that the problem
is coNP-complete.
(B)~For the case of memoryless strategies, as the games are not determined, we 
consider the problem of determining if player 1 has a memoryless 
winning strategy. First, we show that the problem of determining if player 1 
has a memoryless winning strategy is in NP, 
and then show that the problem is NP-hard even when the weights are restricted 
to $\{-1,0,1\}$ and in dimension~$2$.

Finally, we study the computational complexity of multi-mean-payoff games for 
infinite-memory strategies.
Our complexity results are summarized as follows.
(A)~We show that multi-mean-payoff games with mean-payoff-sup objectives
can be decided in NP $\cap$ coNP (in the same complexity as for games with
single mean-payoff objectives). Moreover, we also show that if mean-payoff 
games with single mean-payoff objective can be solved in polynomial time,
then multi-mean-payoff games with mean-payoff-sup objectives can also be
solved in polynomial time.
(B)~Multi-mean-payoff games with mean-payoff-inf objectives are 
coNP-complete.
(C)~Finally, we show that multi-mean-payoff games with combination of 
mean-payoff-sup and mean-payoff-inf objectives are also coNP-complete.

In summary, our results establish optimal computational complexity results
for multi-mean-payoff and multi-energy games under finite-memory, memoryless
and infinite-memory strategies.

% Related works
\smallskip\noindent{\em Related works.}
Mean-payoff games, which are the one-dimension version of our multi-mean-payoff games, 
have been extensively studied starting with the works of 
Ehrenfeucht and Mycielski in~\cite{EM79} where they prove memoryless determinacy 
for these games. Because of memoryless determinacy, it is easy to show that 
the decision problem for mean-payoff games lies in NP~$\cap$~coNP, 
but despite large research efforts, no polynomial time algorithm is known for 
that problem. A pseudo-polynomial time algorithm has been proposed by Zwick 
and Paterson in~\cite{ZP96}, and improved in~\cite{BCDGR09}. The one-dimension
special case of multi-energy games have been introduced in~\cite{CAHS03}
and further studied in~\cite{BFLMS08} where log-space equivalence with 
classical mean-payoff games is established.

%In~\cite{CAHS03}, we have introduced resource interfaces, and Bouyer et al. have 
%introduced weighted (timed) automata and games~\cite{BFLMS08}. 
%In the two papers, the authors have studied models where accumulated weight 
%along runs are subject to constraints. For one important variant of these models, 
%the so-called {\em energy games} (with lower bound constraints), they have proved 
%log-space equivalence to classical mean-payoff games. This log-space equivalence
% allows them to inherit the existing algorithms for solving mean-payoff games. 
%The algorithm proposed in~\cite{BCDGR09} improve the pseudo-polynomial time 
%solutions inherited from mean-payoff games.

%%%%ADD WORK OF ALUR, DEGORRE ON MEAN-PAYOFF AUTOMATA. ADD WORK BY KRISH AND LAURENT ON MEAN-PAYOFF EXPRESSIONS.

%In~\cite{DDGRT10}, we have studied mean-payoff and energy games with imperfect 
%information. When energy games with imperfect information are finite memory 
%determined as our generalized energy games, they are undecidable in general.

Multi-energy games can be viewed as games played on {\sc Vass} (vector addition systems with states)
where the objective
is to avoid unbounded decreasing of the counters. A solution to such games 
on {\sc Vass} is provided in~\cite{BJK10} (see in particular Lemma 3.4 in~\cite{BJK10}) with a PSPACE 
algorithm when the weights are $\{-1,0,1\}$, leading to an EXPSPACE algorithm when the
weights are arbitrary integers.
We drastically improve the EXPSPACE upper-bound by providing a coNP
algorithm for the problem, and we also provide a coNP lower bound even when 
the weights are restricted to $\{-1,0,1\}$.
Finally the work in~\cite{FJLS11} considers multi-dimension energy games with fixed
initial credit, as well as variants of energy games with upper and lower energy bounds.

%\section{Generalized Mean-payoff and Energy Games}\label{sec:def}
\section{Definitions}\label{sec:def}

\noindent{\bf Well quasi-orders.} 
A relation $\preceq$ over a set~$D$ is a {\em well quasi-order} if the following 
conditions hold: $(a)$~$\preceq$ is transitive and reflexive,
and $(b)$~for all $f : \nat \rightarrow D$, there exist
$i_1,i_2 \in \nat$ such that $i_1 < i_2$ and $f(i_1) \preceq f(i_2)$.
It is known that $(\nat^k,\leq)$ is a well quasi-order and that the
Cartesian product of two well quasi-ordered sets is a well quasi-ordered set~\cite{Dickson}.

%\begin{lemma}\label{wqo}
%%Let 
%\end{lemma} 

\smallskip\noindent{\bf Multi-weighted two-player game structures.}
A {\em multi-weighted two-player game structure} (or simply a \emph{game}) is a tuple 
$G=(S_1,S_2,E,w)$ where $S_1 \cap S_2 = \emptyset$, 
and $S_i$ ($i = 1,2$) is the finite set of player-$i$ states (we denote by $S = S_1 \cup S_2$ the state space), 
% $s_{{\sf init}} \in S_1$ is the {\em initial state}, 
$E \subseteq S \times S$ is 
the set of edges such that for all $s \in S$, 
there exists $s' \in S$ such that $(s,s') \in E$, 
% $k \in \nat$ is the {\em dimension} of the multi-weights, 
and $w : E \to \integ^k$ is the multi-weight labeling function. 
The parameter $k \in \nat$ is the {\em dimension} of the multi-weights.
The game $G$ is a {\em one-player} game if $S_2 = \emptyset$.
%, there exists a unique $s' \in S_1$ such that $(s,s') \in E$.
%
The subgraph of $G$ induced by a set $T \subseteq S$ is 
$G \upharpoonright T = (S_1 \cap T,S_2 \cap T, E \cap (T \times T),w)$.
Note that $G \upharpoonright T$ is a game structure if for all $s \in T$, 
there exists $s' \in T$ such that $(s,s') \in E$.

A {\em play} in $G$ from an initial state $s_{{\sf init}} \in S$ is an infinite sequence
$\pi=s_0 s_1 \dots s_n \dots$ of states 
such that $(i)$ $s_0=s_{{\sf init}}$, and $(ii)$ $(s_i,s_{i+1}) \in E$ for all $i \geq 0$. 
The {\em prefix} of length $n$ of $\pi$ is the finite sequence 
$\pi(n)=s_0 s_1 \dots s_n$, its last element $s_n$ is denoted~$\Last(\pi(n))$
and its length~$\abs{\pi(n)}$. 
%A prefix $\pi(n)$ belongs to player~$i$ ($i \in \{1,2\}$) if $\Last(\pi(n)) \in S_i$. 
The set of all plays in $G$ is denoted $\Plays(G)$.
%, 
%the corresponding set of prefixes is denoted by $\Prefs(G)$, 
%the set of prefixes that belongs to player~$i$ ($i \in \{1,2\}$) is denoted 
%by $\Prefs_i(G)$, and the set of ultimately periodic plays in $G$ is denoted by $\Plays^{up}(G)$. 

%%The multi-weight of a prefix $\pi(n)=s_0 s_1 \dots s_n \in \Prefs(G)$ is the 
%%vector of values $v \in \integ^k$ such that $v=\sum_{i=0}^{i=n-1} w(s_i,s_{i+1})$. 

The {\em energy level vector} of a play prefix $\rho=s_0 s_1 \dots s_n$ is 
$\EL(\rho)=\sum_{i=0}^{i=n-1} w(s_i,s_{i+1})$, and the {\em mean-payoff vectors} 
of a play $\pi=s_0 s_1 \dots s_n \dots$ are defined as follows (in dimension $1 \leq j \leq k$):
$\MPsup(\pi)_j = \limsup_{n\to \infty} \frac{1}{n} \cdot \EL(\pi(n))_j$, and 
$\MPinf(\pi)_j = \liminf_{n\to \infty} \frac{1}{n} \cdot \EL(\pi(n))_j$.

\smallskip\noindent{\bf Strategies.}
A \emph{strategy} of player~$i$ ($i \in \{1,2\}$) in~$G$ is a function 
$\lambda_i : S^* \cdot S_i \to S$ such that $(s,\lambda_i(\rho \cdot s)) \in E$
for all $\rho \in S^*$ and all $s \in S_i$. 
A play $\pi=s_0 s_1 \dots \in \Plays(G)$ is \emph{consistent} with a strategy $\lambda_i$ of player~$i$ if 
$s_{j+1}=\lambda_i(s_0 s_1 \dots s_j)$ for all $j \geq 0$ such that $s_j \in S_i$.
The {\em outcome} from a state $s_{{\sf init}}$ of a pair of strategies, 
$\lambda_1$ for player~1 and $\lambda_2$ for player~2, 
is the (unique) play from $s_{{\sf init}}$ that is consistent with both 
$\lambda_1$ and $\lambda_2$. We denote $\outcome_G(s_{{\sf init}},\lambda_1,\lambda_2)$ this play. 
We denote by $T_{\lambda_i(s_{{\sf init}})}$ the \emph{strategy tree} obtained as the unfolding  
of the game $G$ from $s_{{\sf init}}$ when strategy $\lambda_i$ is used. 
The nodes of this tree are all prefixes of the plays from $s_{{\sf init}}$ 
that are consistent with the strategy $\lambda_i$ of player~$i$.

%\mynote{L: todo: check that $\outcome(\cdot,\cdot,\cdot)$ is used with 3 arguments in the sequel.}

%A strategy for player~$i$ is {\em memoryless} if for all $\rho_1,\rho_2 \in  \Prefs_i(G)$, 
%whenever $\Last(\rho_1)=\Last(\rho_2)$ then $\lambda_i(\rho_1)=\lambda_i(\rho_2)$. 
%Note that memoryless strategies of player~$i$ can be seen as functions from $S_i$ 
%into $S_{3-i}$, we often adopt this view in the sequel and abuse notations accordingly.  

A strategy $\lambda_i$ for player~$i$ uses {\em finite-memory} if it can be encoded 
by a deterministic Moore machine $(M,m_0,\alpha_u,\alpha_n)$ where $M$ is a finite 
set of states (the memory of the strategy), $m_0 \in M$ is the initial memory state, 
$\alpha_u: M \times S \to M$ is an update function, 
and $\alpha_n : M \times S_{i} \to S$ is the next-action function. 
% The \emph{size} of the strategy is the number $\abs{M}$ of memory values.
If the game is in a player-$i$ state $s \in S_i$ and $m \in M$ is the current memory value,
then the strategy chooses $s' = \alpha_n(m,s)$ as the next state and the memory 
is updated to $\alpha_u(m,s)$. Formally, $\tuple{M, m_0, \alpha_u, \alpha_n}$
defines the strategy $\lambda$ such that $\lambda(\rho \cdot s) = \alpha_n(\hat{\alpha}_u(m_0, \rho), s)$
for all $\rho \in S^*$ and $s \in S_i$, where $\hat{\alpha}_u$ extends $\alpha_u$ to sequences
of states as usual. The strategy is \emph{memoryless} if $\abs{M} = 1$.
Given an initial state $s_{{\sf init}}$ and a finite-memory strategy $\lambda_i$ of player~$i$, 
let $G_{\lambda_i(s_{{\sf init}})}$ be the graph obtained as the product
of $G$ with the Moore machine defining $\lambda_i$, with initial vertex $\tuple{m_0,s_{{\sf init}}}$ 
and where $(\tuple{m,s},\tuple{m',s'})$ is a transition in the graph if 
$m' = \alpha_u(m,s)$, and either $s \in S_i$ and $s'=\alpha_n(m,s)$, or $s \in S_{3-i}$ and $(s,s') \in E$.
%The set of infinite paths from in $G_{\lambda_i(s_{{\sf init}})}$ 
%and the set of plays consistent with $\lambda_1$ coincide. 
%A similar definition can be given for the case of player~2.

% $\alpha_u: M \times S_{1}$ is an update function, and $\alpha_n : M \rightarrow S_2$ 
% is the next-action function. In state $m \in M$, the strategy plays the state 
% $\alpha_n(m)$ and when the player~$2$ chooses its next state $s \in S_1$, 
% the memory state is updated to $m'=\alpha_u(m,s)$. Formally, the strategy $\lambda_1$ 
% defined by the Moore machine $(M,m_0,\alpha_u,\alpha_n)$ is such that 
% $\lambda_1(\rho)=\alpha_n(\hat{\alpha}_u(m_0,\rho))$ for all $\rho \in \Prefs_1(G)$, 
% where $\hat{\alpha}_u$ extends $\alpha_u$ to sequences of states as follows: 
% $\hat{\alpha}_u(m,\epsilon)=m$, and 
% $\hat{\alpha}_u(m,s_0 s_1 \dots s_n)=\hat{\alpha}_u(\alpha_u(m,s_0),s_1 \dots s_n)$ 
% for all $m$ and $s_0 s_1 \dots s_n \in S_2^{*}$.

\smallskip\noindent{\bf Objectives.}
An {\em objective} for player~$1$ in $G$ is a set of plays $\varphi \subseteq \Plays(G)$. 
Given a game $G$, an initial state $s_0$, and an objective $\varphi$, 
we say that a strategy $\lambda_1$ is {\em winning} for player~1 from $s_0$
if for all plays $\pi \in \Plays(G)$ from $s_0$ that are consistent with $\lambda_1$, 
we have that $\pi \in \varphi$; and we say that a strategy $\lambda_2$ is {\em winning} 
for player~2 from $s_0$ if for all plays in $\pi \in \Plays(G)$ from $s_0$ 
that are consistent with $\lambda_2$, we have that $\pi \not\in \varphi$. 
We denote by $\win{1} \varphi$
the set of states $s_0$ such that there exists a winning strategy for player~$1$ 
from $s_0$, and by $\win{2} \lnot \varphi$ the set of 
states $s_0$ such that there exists a winning strategy for player~$2$ from $s_0$.
Note that $\win{1} \varphi \cap \win{2} \lnot \varphi = \emptyset$ by definition.
We consider the following objectives:

\begin{itemize}
   \item {\em Energy objectives}. Given an initial energy vector 
   $v_0 \in \nat^k$, the {\em multi-energy objective} 
   $\PosEnergy_G(v_0)=\{ \pi \in \Plays(G) \mid \forall n \geq 0 : v_0 + \EL(\pi(n)) \geq \{0\}^k \}$ 
   requires that the energy level in all dimensions remain always nonnegative.
   
%   \item {\em Multi Mean-payoff objectives}. Given a threshold vector 
%   $v \in \integ^k$, the {\em multi mean-payoff objective} 
%   $\MeanPayoffSup_G(v)=\{ \pi \in \Plays(G) \mid \MPsup(\pi) \geq v \}$ 
%   requires for all dimensions $j$ the mean-payoff for dimension $j$ to be at least 
%   $v(j)$.

   \item {\em Mean-payoff objectives}. Given two sets $I,J \subseteq \{1,\dots,k\}$,
   the {\em multi-mean-payoff objective} $\MeanPayoffInfSup_G(I,J) = 
   \{ \pi \in \Plays(G) \mid \forall i\in I: \MPinf(\pi)_i \geq 0 \,\land\, \forall j\in J: \MPsup(\pi)_j \geq 0 \}$ 
   requires for all dimensions in $I$ the mean-payoff-inf value be nonnegative, 
   and for all dimensions in $J$ the mean-payoff-sup value be nonnegative.

 \end{itemize}

When the game $G$ is clear from the context we omit the subscript in objective names.
Note that arbitrary thresholds $\frac{a}{b} \in \rat$ can be considered in the multi-mean-payoff objectives
because the mean-payoff value computed according to the weight function $w$ is greater than $\frac{a}{b}$ 
if and only if the mean-payoff value according to the weight function $b\cdot w - a$ is greater than $0$ 
where $(b \cdot w - a)(e) = b \cdot w(e) - a$ for all $e \in E$.  
For the special case of $I = \emptyset$ and $J = \{1,\dots,k\}$, 
we denote by $\MeanPayoffSup = \MeanPayoffInfSup(\emptyset,J)$ the conjunction
of all mean-payoff-sup objectives, and for $I = \{1,\dots,k\}$ and $J = \emptyset$ we 
denote by $\MeanPayoffInf = \MeanPayoffInfSup(I,\emptyset)$ 
the conjunction of all mean-payoff-inf objectives. We denote by 
$\MeanPayoffSup_i = \MeanPayoffInfSup(\emptyset,\{i\})$ the single 
mean-payoff-sup objective in dimension $1 \leq i \leq k$.

\smallskip\noindent{\bf Decision problems.}
We consider the following decision problems:
\begin{itemize}

  \item The {\em unknown initial credit problem} asks, given a
  multi-weighted two-player game structure $G$, and an initial state $s_0$, 
  to decide whether there exist an initial credit vector $v_0 \in \nat^k$ 
  and a winning strategy $\lambda_1$ for player~1 from $s_0$ for the objective $\PosEnergy_G(v_0)$.

  \item The {\em mean-payoff threshold problem} asks, given a 
  multi-weighted two-player game structure $G$, an initial state $s_0$, and 
  two sets $I,J \subseteq \{1,\dots,k\}$ of indices,
  %and a threshold vector $v \in \mathbb{Z}^k$, 
  to decide whether there exists a winning strategy $\lambda_1$ for player~1 
  from $s_0$ for the objective $\MeanPayoffInfSup_G(I,J)$.

\end{itemize}

% \noindent Note that in the unknown initial credit problem, we allow arbitrary strategies
% (and we show in Theorem~\ref{thrm_gen_energy_fin} that actually finite-memory strategies  are sufficient),
% while in the mean-payoff threshold problem, we require finite-memory strategy
% which is restriction (according to Theorem~\ref{thrm_gen_mean}) of a more general problem 
% of deciding the existence of arbitrary winning strategies.

\smallskip\noindent{\bf Determinacy, determinacy under finite-memory, and determinacy by finite-memory.}
We now define the notion of determinacy, determinacy under finite-memory and determinacy by finite-memory.
\begin{itemize}
\item \emph{(Determinacy).} 
A game $G$ with state space~$S$ and objective~$\varphi$ is \emph{determined} if from all states $s_0 \in S$,
either player~1 or player~2 has a winning strategy, i.e. $S = \win{1} \varphi \cup \win{2} \lnot \varphi$.
Observe that since $\win{1}\varphi \cap \win{2} \lnot \varphi=\emptyset$, determinacy 
means that $\win{1} \varphi$ and $\win{2} \lnot \varphi$ partition the state space.

\item \emph{(Determinacy under finite-memory).}
We also consider \emph{determinacy under finite-memory} strategies. 
Let $\win{1}^{finite} \varphi$ be the set of states $s_0$ from which player~1 has a \emph{finite-memory} strategy $\lambda_1$ such that for all 
\emph{finite-memory} strategies $\lambda_2$ of player~2, we have $\outcome_G(s_0, \lambda_1,\lambda_2) \in \varphi$.
And let $\win{2}^{finite} \lnot \varphi$ be the set of states $s_0$ from which player~1 has a \emph{finite-memory} 
strategy $\lambda_2$ such that for all \emph{finite-memory} strategies $\lambda_1$ of player~1, 
we have $\outcome_G(s_0, \lambda_1,\lambda_2) \not\in \varphi$.
A game $G$ with state space~$S$ and objective~$\varphi$ is \emph{determined under finite-memory} if 
$S = \win{1}^{finite} \varphi \cup \win{2}^{finite} \lnot \varphi$.
Again observe that $\win{1}^{finite}\varphi \cap \win{2}^{finite} \lnot \varphi=\emptyset$, 
and determinacy under finite-memory means that 
$\win{1}^{finite} \varphi$ and $\win{2}^{finite} \lnot \varphi$ partition the state space.
We say that determinacy and determinacy under finite-memory coincide for an objective~$\varphi$, 
if for all game structures, we have $\win{1} \varphi = \win{1}^{finite} \varphi$
and $\win{2} \lnot \varphi = \win{2}^{finite} \lnot \varphi$.

\item \emph{(Determinacy by finite-memory).}
We also consider \emph{determinacy by finite-memory} strategies.
Let $\win{1}^{fin-inf}\varphi$ be the set of states $s_0$ from which player~1 has 
a \emph{finite-memory} strategy $\lambda_1$ such that for all strategies $\lambda_2$ of player~2, 
we have $\outcome_G(s_0, \lambda_1,\lambda_2) \in \varphi$ (i.e., player~1 is restricted to 
finite-memory strategies whereas strategies for player~2 are general infinite-memory strategies).
The set of states $s_0$ from which player~2 has 
a \emph{finite-memory} strategy $\lambda_2$ such that for all strategies $\lambda_1$ of player~1, 
we have $\outcome_G(s_0, \lambda_1,\lambda_2) \not\in \varphi$ is denoted $\win{2}^{fin-inf}\lnot \varphi$.
%%Note that in general we have $\win{1}^{fin-inf}\varphi \subseteq \win{1}\varphi$ and 
%%$\win{1}^{fin-inf}\varphi \subseteq \win{1}^{finite}\varphi$, and we have similar 
%%inclusion relation for player~2.
If for all game structures we have $\win{1}\varphi=\win{1}^{fin-inf}\varphi$ and 
$\win{2}\lnot \varphi =\win{2}^{fin-inf}\lnot \varphi$, and all game
structures with objective~$\varphi$ are determined, then 
we say that determinacy by finite-memory strategies holds for $\varphi$.
\end{itemize}

We first observe that determinacy by finite-memory strategies implies that 
finite-memory strategies suffice for both players, and determinacy by 
finite-memory implies determinacy under finite-memory (since given a 
finite-memory strategy of a player, if there is a counter strategy 
for the opponent, then there is a finite-memory one by determinacy by 
finite-memory).
Thus determinacy by finite-memory strategies implies that 
(i)~$\win{1}\varphi =\win{1}^{finite}\varphi =\win{1}^{fin-inf}\varphi$; and
(ii)~$\win{2}\lnot\varphi =\win{2}^{finite}\lnot\varphi =\win{2}^{fin-inf}\lnot\varphi$.
As we will show that determinacy and determinacy under finite-memory do not coincide 
for multi-mean-payoff games (Theorem~\ref{thrm_gen_mean}), we consider for 
multi-mean-payoff objectives $\varphi$ both 
(1)~winning under finite-memory strategies, i.e. to decide whether 
$s_0 \in \win{1}^{finite}\varphi$ for a given initial state $s_0$; and
(2)~winning under general strategies, i.e. to decide whether 
$s_0 \in \win{1}\varphi$ for a given initial state $s_0$.
For multi-energy games we will show determinacy by finite-memory strategies.

%if for all objectives in the class and all game structures,  
%the answer of the determinacy and determined under finite-memory 
%coincide (i.e., player~1 has a winning strategy iff there is a finite-memory 
%winning strategy, and similarly for player~2). 

Determinacy for multi-mean-payoff and multi-energy objectives follows from 
a general determinacy result for Borel objectives~\cite{Mar75}:
(a)~multi-mean-payoff objectives can be expressed as a finite intersection 
of one-dimensional mean-payoff objectives which are complete for the third level of the 
Borel hierarchy~\cite{ChaTCS}; and 
(b)~multi-energy objectives can be expressed as a finite intersection 
of one-dimensional energy objectives which are closed sets.
%Hence determinacy of generalized mean-payoff and energy games follows 
%from the result of~\cite{Mar75}.

\begin{theorem}[Determinacy~\cite{Mar75}]
Multi-mean-payoff and multi-energy games are determined. 
\end{theorem}
%%%%%%

\smallskip\noindent{\bf Attractors.} 
The player-$1$ \emph{attractor} of a given set $T \subseteq S$ of target states 
is the set of states from which player~$1$ can force to eventually reach a state in $T$. 
The attractor is defined inductively as follows:
let $A_0 = T$, and for all $j \geq 0$ let 
\[
A_{j+1} = A_j \cup \{s \in S_1 \mid \exists(s,t) \in E:  t \in A_j\} 
\cup \{s \in S_2 \mid \forall (s,t)\in E:  t \in A_j\}
\]
denote the set of states from where player~$1$ can ensure to reach $A_j$ 
within one step irrespective of the choice of player~$2$. 
Then the player-$1$ attractor is $\AttrOne{T}=\bigcup_{j\geq 0} A_j$.
The player-$2$ \emph{attractor} $\AttrTwo{T}$ is defined symmetrically. 
%In sequel we only use player-2 attractors.
Note that for $i=1,2$, the subgraph $G \upharpoonright (S \setminus \Attri{T})$ 
is again a game structure (i.e., every state has an outgoing edge). 
%%The attractors can be computed in linear time~\cite{}.
%
For all multi-mean-payoff objectives $\varphi$ (and in general for all tail objectives~\cite{ChaTCS}),
we have $\win{1} \varphi = \AttrOne{\win{1} \varphi}$ and 
$\win{2} \lnot \varphi = \AttrTwo{\win{2} \lnot \varphi}$.

%\section{Determinacy under Finite-memory and Inter-reducibility}
\section{Multi-Energy Games}\label{sec:multi-energy}

In this section, we study the determinacy and complexity of multi-energy games.
First, we show that \emph{finite-memory} strategies are sufficient for player~$1$,
and \emph{memoryless} strategies are sufficient for player~$2$. It follows that
multi-energy games are determined under finite-memory. % (Theorem~\ref{thrm_gen_energy_fin}).
We establish coNP complexity for the unknown initial credit problem,
as well as a matching coNP-hardness result, 
and we show that under memoryless strategies for player~$1$ the problem is NP-complete.
Finally, we show that the unknown initial credit problem is log-space equivalent
to the mean-payoff threshold problem when the players have to use finite-memory
strategies (and in general infinite-memory strategies are more powerful than 
finite-memory strategies in multi-mean-payoff games). The case of 
infinite-memory strategies in multi-mean-payoff games is addressed in Section~\ref{sec:multi-mean-payoff}.

\paragraph{\bf Determinacy under finite-memory.}
The next lemmas show that finite-memory strategies are sufficient for player~1
in multi-energy games, and that memoryless strategies are sufficient for player~2.

\begin{lemma}\label{lem:energy-player1-finite-memory}
For all multi-weighted two-player game structures $G$ and initial states $s_0$, 
the answer to the unknown initial credit problem is {\sc Yes} if and only if 
there exist an initial credit $v_0 \in \nat^k$ and a \emph{finite-memory} strategy $\lambda^{\sf FM}_1$ 
for player~$1$ such that for all strategies $\lambda_2$ of player~$2$, 
$\outcome_G(s_0, \lambda^{\sf FM}_1,\lambda_2) \in \PosEnergy_G(v_0)$.
\end{lemma}

\begin{proof}  %%\emph{(of Lemma~\ref{lem:energy-player1-finite-memory}).}
One direction is trivial. For the other direction, assume that $\lambda_1$ is a 
(not necessary finite-memory) winning strategy for player~1 in $G$ from $s_0$ with initial 
credit $v_0 \in \nat^k$. We show how to construct from $\lambda_1$ 
a finite-memory strategy $\lambda_1^{\sf FM}$ that is winning from $s_0$ against all 
strategies of player~2 for initial credit $v_0$. 

Consider the strategy tree $T_{\lambda_1(s_{0})}$ and associate to each node 
$\rho = s_0 s_1 \dots s_n$ in this tree the energy vector $v_0 + \EL(\rho)$. 
Since $\lambda_1$ is winning, we have $v_0 + \EL(\rho) \in \nat^k$ for all $\rho \in T_{\lambda_1(s_{0})}$. 
Now, consider the relation $\sqsubseteq$ on the set $S \times \nat^k$
defined as follows: $(s_1,v_1) \sqsubseteq (s_2,v_2)$ if $s_1 = s_2$ and 
$v_1 \leq v_2$ (i.e., $v_1(i) \leq v_2(i)$ for all $i$, $1 \leq i \leq k$). 
The relation $\sqsubseteq$ is a well quasi-order. 
As a consequence, on every infinite branch $\pi = s_0 s_1 \dots s_n \dots$ 
of $T_{\lambda_1(s_{0})}$ there exist two indices $i < j$ such that 
$\Last(\pi(i)) = \Last(\pi(j))$ and $\EL(\pi(i)) \leq \EL(\pi(j))$.
We say that node $\pi(j)$ subsumes node $\pi(i)$. 
Now, let $T^{\sf FM}$ be the tree $T_{\lambda_1(s_{0})}$ where we stop each branch 
when we reach a node $n_2$ that subsumes one of its ancestor node $n_1$. 
By K\"{o}nig's lemma~\cite{Konig36} and Dickson's lemma~\cite{Dickson}, the tree
$T^{\sf FM}$ is finite. From the node $n_2$, player~1 can mimic the strategy played
in $n_1$ because the energy level in $n_2$ is greater than in $n_1$.  
From $T^{\sf FM}$, we can construct the Moore machine of a finite-memory 
strategy $\lambda_1^{\sf FM}$ that is winning in the multi-energy game~$G$ 
from~$s_0$ with initial energy level~$v_0$.
\hfill\qed
\end{proof}

\begin{lemma}[\cite{BJK10}]\label{lem:player-two-memoryless}
%Memoryless strategies are sufficient for player~$2$ in multi energy games.
For all multi-weighted two-player game structures $G$ and initial states $s_0$, 
the answer to the unknown initial credit problem is {\sc No} if and only if 
there exists a \emph{memoryless} strategy $\lambda_2$ for player~$2$, such that
for all initial credit vectors $v_0 \in \nat^k$ and all strategies $\lambda_1$ 
for player~$1$ we have $\outcome_G(s_0,\lambda_1,\lambda_2) \not\in \PosEnergy_G(v_0)$.
\end{lemma}

\begin{proof}  %%\emph{(of Lemma~\ref{lem:player-two-memoryless}).}
The proof was given in~\cite[Lemma 19]{BJK10}. Intuitively, consider a player-$2$ 
state $s \in S_2$ with two successors $s'$ and $s''$. 
If an initial credit vector $v'_0$ is sufficient for player~$1$ to win from $s_{\sf init}$
against player~$2$ always choosing $s'$, and $v''_0$ is sufficient from $s$ against  
player~$2$ always choosing $s''$, then $v'_0+v''_0$ is sufficient from $s_{\sf init}$ against 
player~$2$ arbitrarily alternating between $s'$ and $s''$. This is because
if player~$1$ maintains the energy nonnegative in all dimensions when the initial credit
is $v_0$, then he can maintain the energy always above $\Delta$ when initial credit
is $v_0 + \Delta$ ($\Delta \in \nat^k$).
\hfill\qed
%\mynote{L: should we say more ?}
\end{proof}

%%The proof of Lemma~\ref{lem:player-two-memoryless} is in the appendix.
\noindent The previous two lemmas establishes both determinacy by finite-memory 
strategies, as well as that determinacy and determinacy under finite-memory 
coincide.
As a consequence of the previous two lemmas, we get the following theorem.

\begin{theorem}\label{thrm_gen_energy_fin}
Multi-energy games are determined by finite-memory, hence determined under finite-memory. 
Determinacy coincides with determinacy under finite-memory for multi-energy games.
\end{theorem}

\begin{remark}
Note that even if player 2 can be restricted to play memoryless strategies in 
multi-energy games, it may be that player~$1$ is winning with some initial 
credit vector $v_0$ when player~$2$ is memoryless, and is not winning with 
the same initial credit vector $v_0$ when player~$2$ can use arbitrary strategies.
This situation is illustrated in \figurename~\ref{fig:memory-needed} where player~$1$
(owning round states) can maintain the energy nonnegative in all dimensions 
with initial credit $(2,0)$ when player~$2$ (owning square states) is memoryless.
Indeed, either player~$2$ chooses the left edge from $s_0$ to $s_1$ and player~$1$ wins, 
or player~$2$ chooses the right edge from $s_0$ to $s_2$, and player~$1$ wins as well by 
alternating the edges back to $s_0$. Now, if player~$2$ has memory, then player~2 wins
by choosing first the right edge to $s_2$, which forces player~$1$ to come back
to $s_0$ with multi-weight $(-1,1)$. The energy level is now $(1,1)$ in $s_0$ and player~$2$
chooses the left edge to $s_1$ which is losing for player~$1$. Note that player~$1$
wins with initial credit $(2,1)$ and $(3,0)$ (or any larger credit) against all 
arbitrary strategies of player~$2$.
%\mynote{L: should I put this example ?}
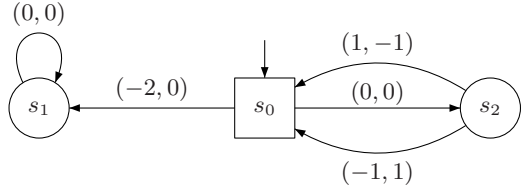
\begin{figure}[!tb]
 \begin{center}
   %%\hrule

\begin{picture}(75,28)(0,0)
%\put(0,0){\framebox(75,28){}}

%\gasset{Nw=9,Nh=9,Nmr=4.5,rdist=1, loopdiam=6}

\node[Nmarks=i, iangle=90, Nmr=0](n0)(40,12){$s_0$}
\node[Nmarks=n](n1)(10,12){$s_1$}
\node[Nmarks=n](n2)(70,12){$s_2$}

\drawloop[ELside=l,loopCW=y, loopdiam=6](n1){$(0,0)$}

%\drawloop[ELside=l,loopCW=y](nk){$0,1$}

\drawedge[ELpos=50, ELside=l, ELdist=0.5, curvedepth=0](n0,n2){$(0,0)$}
\drawedge[ELpos=50, ELside=l, curvedepth=6](n2,n0){$(-1,1)$}
\drawedge[ELpos=50, ELside=r, curvedepth=-6](n2,n0){$(1,-1)$}

\drawedge[ELpos=50, ELside=r, curvedepth=0](n0,n1){$(-2,0)$}

%\drawedge[dash={1}0](n3bis,nkbis){$0,1$}

\end{picture}
   %%\hrule
   \caption{player~$1$ (round states) wins with initial credit $(2,0)$ when player~$2$ (square states) can use memoryless strategies, 
but not when player~$2$ can use arbitrary strategies.   \label{fig:memory-needed}}
 \end{center}
\end{figure}
\end{remark}

\paragraph{\bf Complexity.}
%\section{coNP-completeness for Finite-Memory Strategies}
We show that the unknown initial credit problem is coNP-complete.
%In this section, we present a nondeterministic polynomial time algorithm 
%that makes one call to a nondeterministic polynomial time oracle 
%to recognize the instances for which there is no winning strategies for player~1 
%in a multi-energy game. 
First, we show that the one-player version of this game can be solved 
by checking the existence of a circuit (i.e., a not necessarily simple cycle) 
with nonnegative effect in all dimensions, and we use the memoryless result 
for player~2 (Lemma~\ref{lem:player-two-memoryless}) to define a coNP algorithm. 
Second, we present a coNP-hardness proof.

%The main result (Theorem~\ref{thrm_complete}) is derived from 
%Lemma~\ref{lem:coNp-membership} and Lemma~\ref{thrm_hard} below.

%($\Pi^2$ is the second level of the polynomial hierarchy and represents the 
%class of algorithms described as a coNP algorithm with a NP oracle).

\begin{theorem}\label{thrm_complete}
The unknown initial credit problem is coNP-complete. 
\end{theorem}

%\paragraph{coNP upper bound}
%\smallskip\noindent{\bf coNP upper bound.}

First, we need the following result about zero-circuits in 
multi-weighted directed graphs (a graph is a one-player game).
A \emph{zero-circuit} is a finite sequence $s_0 s_1 \dots s_n$ with $n \geq 1$ 
such that $s_0 = s_n$, $(s_i,s_{i+1}) \in E$ for all $0 \leq i < n$, and 
$\sum_{i=0}^{n-1} w(s_i,s_{i+1}) = (0,0,\dots,0)$. 
The circuit need not be simple.

\begin{lemma}[\cite{KS88}]\label{lem:zero-cycle}
Deciding if a multi-weighted directed graph contains a zero circuit
can be done in polynomial time.
\end{lemma}

\noindent The result of Theorem~\ref{thrm_complete} follows from the next two lemmas.

% \begin{lemma}
% \label{1PG-Pmembership}
% The unknown initial credit problem for a multi-weighted one-player game 
% structure $G$ is solvable polynomial time.
% \end{lemma}

\begin{lemma}\label{lem:coNp-membership}
The unknown initial credit problem is in coNP. %$\Pi^2$ (in ${\sc coNP}^{{\sc NP}}$). 
\end{lemma}

\begin{proof}
Let $G$ be a multi-weighted two-player game structure, and $s_0$ be an initial state.
By Lemma~\ref{lem:player-two-memoryless}, we know that player~2 can be restricted 
to play memoryless strategies. A coNP algorithm guesses a memoryless strategy~$\lambda_2$ 
and checks in polynomial time that it is winning for player~$2$ using the following argument.

Consider the graph $G_{\lambda_2(s_0)}$ as a one-player game (in which all states
belong to player~$1$). We show that if there exists an initial energy level $v_0$ and an 
infinite play $\pi=s_0 s_1 \dots s_n \dots$ in $G_{\lambda_2(s_0)}$ such that $\pi \in \PosEnergy(v_0)$,
then there exists a reachable circuit in $G_{\lambda_2(s_0)}$ with nonnegative effect in all 
dimensions. To show this, we extend $\pi$ with the energy level as follows: let
$\pi'=(s_0,w_0) (s_1,w_1) \dots (s_n,w_n) \dots$ where $w_0=v_0$ and for all 
$i \geq 1$, $w_i=v_0+\EL(\pi(i))$. Since $\pi \in \PosEnergy(v_0)$, we know that 
$w_i \in \nat^k$ for all $i \geq 0$. Hence the following order defined 
on the pairs $(s,w) \in S \times \nat^k$ is a well quasi-order: 
$(s,w) \sqsubseteq (s',w')$ if $s=s'$ and $w(j) \leq w'(j)$ for all $1 \leq j \leq k$.
%It is easy to show that $\sqsubseteq$ is a well quasi-order. 
It follows that there exist two indices $i_1 < i_2$ in $\pi'$ such that 
$(s_{i_1},w_{i_1}) \sqsubseteq (s_{i_2},w_{i_2})$, and the underlying circuit 
through $s_{i_1} = s_{i_2}$ has nonnegative effect in all dimensions.

Based on this, we can decide if there exists an initial energy vector $v_0$ and an infinite path 
in $G_{\lambda_2(s_0)}$ that satisfies $\PosEnergy_G(v_0)$ using the result of Lemma~\ref{lem:zero-cycle} 
on modified version of $G_{\lambda_2(s_0)}$ obtained as follows. 
In every state of $G_{\lambda_2(s_0)}$, we add $k$ self-loops with respective multi-weight $(-1,0,\dots,0)$,
$(0,-1,0,\dots,0)$, $\dots$, $(0,\dots,0,-1)$, i.e. each self-loop removes one unit
of energy in one dimension. It is easy to see that $G_{\lambda_2(s_0)}$ has a circuit with nonnegative 
effect in all dimensions if and only if the modified $G_{\lambda_2(s_0)}$ has a zero circuit, which 
can be determined in polynomial time. The result follows.
\hfill\qed
\end{proof}

\begin{lemma}\label{thrm_hard}
The unknown initial credit problem is coNP-hard. 
\end{lemma}

%\paragraph{Lower bound: coNP-hardness}
%\smallskip\noindent{\bf Lower bound: coNP-hardness.}

\begin{proof}
%We show that the unknown initial credit problem for 
%multi-weighted two-player game structures is coNP-hard. 
We present a reduction from the complement of the 3SAT problem 
which is NP-complete~\cite{PapaBook}.

\smallskip\noindent\emph{Reduction.}
We show that the unknown initial credit problem for multi-weighted two-player 
game structures is at least as hard as deciding whether a 3SAT formula is unsatisfiable.
Consider a 3SAT formula $\psi$ in CNF with clauses $C_1,C_2,\ldots,C_k$ 
over variables $\{x_1, x_2, \ldots, x_n\}$, where each clause consists 
of disjunctions of exactly three literals (a literal is a variable or its
complement). 
Given the formula $\psi$, we construct a game graph as shown in 
Figure~\ref{fig:3sat}.  
The game graph is as follows: from the initial state, player~1 chooses 
a clause, then from a clause player~2 chooses a literal  that appears 
in the clause (i.e., makes the clause true).  From every literal 
the next state is the initial state.
We now describe the multi-weight labeling function $w$.
In the multi-weight function there is a component for every literal.
For edges from the initial state to the clause states, and from 
the clause states to the literals, the weight for every component 
is~0. 
We now define the weight function for the edges from literals back to the 
initial state: for a literal $y$, and the edge from $y$ to the 
initial state, the weight for the component of $y$ is~$1$, the weight for
the component of the complement of $y$ is~$-1$, and for all the other 
components the weight is~$0$.
We now define a few notations related to assignments of truth values 
to literals.
We consider \emph{assignments} that assign truth values to all the literals.
An assignment is \emph{valid} if for every literal the truth value assigned
to the literal and its complement are complementary (i.e., for all $1 \leq i 
\leq n$, if 
$x_i$ is assigned true (resp. false), then the complement $\overline{x}_i$ 
of $x_i$ is assigned false (resp. true)).
An assignment that is not valid is \emph{conflicting} (i.e., for some 
$1 \leq i \leq n$, both $x_i$ and $\overline{x}_i$ are assigned the same 
truth value). 
If the formula $\psi$ is satisfiable, then there is a valid assignment 
that satisfies all the clauses.
If the formula $\psi$ is not satisfiable, then every assignment that satisfies
all the clauses must be conflicting. 
We now present two directions of the hardness proof.

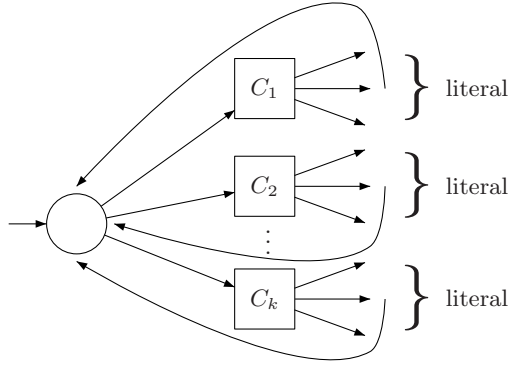
\begin{figure}[tb]
 \centering
 \begin{picture}(65,40)(0,0)

   \drawcurve(51,33)(47,44)(10,20)
   \drawcurve(51,20)(49,12)(15,15)
   \drawcurve(51,5)(49,-2)(10,10)

   \node[Nmarks=i](q0)(10,15){}
   \node[Nmr=0](q1)(35,33){$C_1$}
   \fmark[fangle=20,flength=10](q1)
   \fmark[fangle=0,flength=10](q1)
   \fmark[fangle=-20,flength=10](q1) 
   \node[Nmr=0](q2)(35,20){$C_2$}
   \fmark[fangle=20,flength=10](q2)
   \fmark[fangle=0,flength=10](q2)
   \fmark[fangle=-20,flength=10](q2) 
   \node[Nmr=0](q3)(35,5){$C_k$}
   \fmark[fangle=20,flength=10](q3)
   \fmark[fangle=0,flength=10](q3)
   \fmark[fangle=-20,flength=10](q3) 

   \drawedge(q0,q1){ }
   \drawedge(q0,q2){ }
   \drawedge(q0,q3){ }

   \put(35,11){$\vdots$}
   \put(53,31){{\Huge $\}$}}
   \put(53,18){{\Huge $\}$}}
   \put(53,3){{\Huge $\}$}}
   \put(58,32){ literal}
   \put(58,19){ literal}
   \put(58,4){ literal}

 \end{picture}
 \caption{Game graph construction for a 3SAT formula (Lemma~\ref{thrm_hard}).}
 \label{fig:3sat}
\end{figure}

\smallskip\noindent\emph{$\psi$ satisfiable implies player~2 winning.}
We show that if $\psi$ is satisfiable, then player~2 has a memoryless
winning strategy.  
Since $\psi$ is satisfiable, there is a valid assignment $A$ that 
satisfies every clause.
The memoryless strategy is constructed from the assignment $A$ as follows:
for a clause $C_i$, the strategy chooses a literal as successor that appears 
in $C_i$ and is set to true by the assignment. 
Consider an arbitrary strategy for player~1, and the infinite play: 
the literals visited in the play are all assigned truth values true by $A$,
and the infinite play must visit some literal infinitely often. 
Consider the literal $x$ that appears infinitely often in the play, then the 
complement literal $\overline{x}$ is never visited, and every 
time literal $x$ is visited, the component corresponding to $\overline{x}$ 
decreases by~$1$, and since $x$ appears infinitely often it follows that
the play is winning for player~2 for every finite initial credit.  
It follows that the strategy for player~2 is winning, and the answer to the 
unknown initial credit problem is ``No".

\smallskip\noindent\emph{$\psi$ not satisfiable implies player~1 is winning.} 
We now show that
if $\psi$ is not satisfiable, then player~1 is winning. 
By determinacy, it suffices to show that player~2 is not winning, 
and by existence of memoryless winning strategy for player~2 
(Lemma~\ref{lem:player-two-memoryless}), it suffices to show that there is no 
memoryless winning strategy for player~2. 
Fix an arbitrary memoryless strategy for player~2, (i.e., in every clause
player~2 chooses a literal that appears in the clause). 
If we consider the assignment $A$ obtained from the memoryless strategy, then
since $\psi$ is not satisfiable it follows that the assignment $A$ is 
conflicting.
Hence there must exist clause $C_i$ and $C_j$ and variable $x_k$ such that the 
strategy chooses the literal $x_k$ in $C_i$ and the complement variable 
$\overline{x}_k$ in $C_j$. 
The strategy for player~1 that at the starting state alternates between 
clause $C_i$ and $C_j$, along with that the initial credit of $1$ for the 
component of $x_k$ and $\overline{x}_k$, and~$0$ for all other components, 
ensures that the strategy for player~2 is not winning. 
Hence the answer to the unknown initial credit problem is {\sc Yes}, and we have 
the desired result.
\hfill\qed
\end{proof}

%\begin{remark}
\noindent Observe that our hardness proof works with weights restricted to the 
set $\{-1,0,1\}$. The results of~\cite{Cha10} show that in two dimensions ($k=2$)
the unknown initial credit problem with weights in $\{-1,0,1\}$ can be solved in 
polynomial time. The complexity for fixed dimensions $k \geq 3$ is not known.
With arbitrary integer weights, the unknown initial credit problem for $k=1$ 
is in UP~$\cap$~coUP~\cite{BFLMS08}.
%\end{remark}

\paragraph{\bf Complexity for memoryless strategies.}
%\section{NP-completeness for Memoryless Strategies}
We consider multi-energy games when player~1 is restricted to use 
memoryless strategies. The unknown initial credit problem for memoryless strategies
is to decide, given a multi-weighted two-player game structure $G$, and an initial state $s_0$, 
whether there exist an initial credit vector $v_0 \in \nat^k$ 
and a \emph{memoryless} winning strategy $\lambda_1$ for player~1 
from $s_0$ for the objective $\PosEnergy_G(v_0)$.

\begin{theorem}\label{lemm_memless_1}
The unknown initial credit problem for memoryless strategies is NP-complete.
\end{theorem}

\begin{proof}  %% \emph{(of Lemma~\ref{lemm_memless_1}).}
The inclusion in NP is obtained as follows: the polynomial witness is the
memoryless strategy for player~1, and once the strategy is fixed we obtain 
a game graph with choices for player~2 only. 
The verification is to checks that for every 
dimension there is no negative cycle, and it can be achieved in polynomial time by
solving one-dimensional energy games on graphs with choices for player~2 only~\cite{CAHS03,BFLMS08}.
%The NP upper bound follows.

The NP hardness follows from a result of~\cite{FHW80} where, given a directed graph 
and four vertices $w,x,y,z$, the problem of deciding the existence of two disjoint 
simple paths (one from $w$ to $x$ and the other from $y$ to $z$) is shown to be NP-complete. 
Given such a graph and vertices, construct a one-player game by $(1)$ adding the edges 
$(x,y)$ with weight $(n,-1)$ and $(z,w)$ with weight $(-1,n)$ (where $n$ is the 
number of vertices in the graph), and $(2)$ assigning all other edges of the graph the weight $(-1,-1)$. 
In the resulting one-player game, a winning memoryless strategy from $w$ must induce
a simple cycle through $w,x,y,z$ to ensure nonnegative sum of weights in 
the two dimensions. This show that the unknown initial credit problem for memoryless strategies
is at least as hard as the decision problem of~\cite{FHW80}, and thus NP-hard.
The NP-completeness result follows.
\hfill\qed
\end{proof}

%\begin{remark}
The reduction in the proof of Theorem~\ref{lemm_memless_1} can be obtained
with weights in $\{-1,0,1\}$ by replacing the edges with weight $n$ by 
a sequence of $n$ edges with weight $1$. The reduction remains polynomial.
Theorem~\ref{lemm_memless_1} shows NP-hardness for dimension $k=2$ 
and weights in $\{-1,0,1\}$. For $k=1$, the problem is solvable in polynomial 
time with weights in $\{-1,0,1\}$, and for arbitrary integer weights, 
the problem is in UP $\cap$ coUP~\cite{BFLMS08,BCDGR09}.
%(so the hardness result cannot be obtained for $k=1$). 
%%If $k=2$ and the weights in $\{-1,0,1\}$, then the memoryless strategy problem 
%%can be solved in polynomial time.
%%\mynote{Krish: to discuss to put $k$ as constant.} 
%\end{remark}

\paragraph{\bf Equivalence with multi-mean-payoff games under finite-memory strategies.}
We show that multi-mean-payoff games where the players are restricted 
to play finite-memory strategies are log-space equivalent to multi-energy games. 
The result of Lemma~\ref{thrm_inter}
shows that the unknown initial credit problem (for multi-energy games) and
the mean-payoff threshold problem (with finite-memory strategies) are equivalent.

Note that if the players use finite-memory strategies, then the outcome $\pi$ is 
ultimately periodic (a play $\pi=s_0 s_1 \dots s_n \dots$ is {\em ultimately periodic} if 
it can be decomposed as $\pi=\rho_1 \cdot \rho_2^{\omega}$ where $\rho_1$ and 
$\rho_2$ are two finite sequences of states) and therefore, the value of $\MPsup(\pi)$
and $\MPinf(\pi)$ coincide. We denote by $\MeanPayoff_G$ the set of
ultimately periodic plays satisfying the multi-mean-payoff objective
$\MeanPayoffInf_G$ (or equivalently, satisfying $\MeanPayoffSup_G$).

\begin{lemma}\label{thrm_inter}
\label{lem:energy-mean-payoff-reduction}
For all multi-weighted two-player game structures, the answer to the 
unknown initial credit problem is {\sc Yes} if and only if the answer to the 
mean-payoff threshold problem under finite-memory strategies is {\sc Yes}.
\end{lemma}

\begin{proof}  %%%\emph{(of Theorem~\ref{thrm_inter}).}
Let $G$ be multi-weighted two-player game structure of dimension $k$.
First, assume that there exists a winning strategy $\lambda_1$ for player~$1$
in~$G$ for the energy objective $\PosEnergy_G(v_0)$ (for some $v_0$). 
Theorem~\ref{thrm_gen_energy_fin} establishes that finite memory is sufficient to win
multi-energy games, so we can assume that $\lambda_1$ has finite memory.
Consider the restriction of the graph $G_{\lambda_1}$ to the reachable vertices,
and we show that the energy vector of every simple cycle is nonnegative. By contradiction,
if there exists a simple cycle with energy vector negative in one dimension,
then the infinite path that reaches this cycle and loops through it forever
would violate the objective $\PosEnergy_G(v_0)$ regardless of the vector $v_0$.
Now, this shows that every reachable cycle in $G_{\lambda_1}$ has nonnegative
mean-payoff value in all dimensions, hence $\lambda_1$ is winning for the 
multi-mean-payoff objective $\MeanPayoff_G$.

Second, assume that there exists a finite-memory strategy $\lambda_1$ for player~$1$
that is winning in~$G$ for the multi-mean-payoff objective $\MeanPayoff_G$. 
By the same argument as above, all simple cycles in $G_{\lambda_1}$ are nonnegative 
and the strategy $\lambda_1$ is also winning for the objective $\PosEnergy_G(v_0)$ 
for some $v_0$. Taking $v_0 = \{n W\}^k$ where $n$ is the number of states 
in~$G_{\lambda_1}$ (which bounds the length of the acyclic paths) and $W \in \integers$ 
is the largest weight in the game suffices.
\hfill\qed
\end{proof}

Note that the result of Lemma~\ref{thrm_inter} does not hold for arbitrary
strategies as shown in the following lemma.

\begin{lemma}\label{lemm_inf_power}
In multi-mean-payoff games, in general infinite-memory strategies are required 
for winning (i.e., in general, finite-memory strategies are not sufficient 
for winning). 
\end{lemma}

%\mynote{L: not sure which proof we should keep.}

\begin{proof} %%\emph{(of Lemma~\ref{lemm_inf_power}).} 
% To show this, we first need to define
% the mean-payoff vector of arbitrary plays (because arbitrary strategies, i.e.,
% infinite-memory strategies, may produce non-ultimately periodic plays). In particular, the limit
% of $\frac{1}{n} \cdot \EL(\pi(n))$ for $n \to \infty$ may not exist for arbitrary plays~$\pi$.
% Therefore, two possible definitions are usually considered, namely either
% $\underline{\MP}(\pi) = \liminf_{n \to \infty} \frac{1}{n} \cdot \EL(\pi(n))$,
% or $\overline{\MP}(\pi) = \limsup_{n \to \infty} \frac{1}{n} \cdot \EL(\pi(n))$.
% In both cases, better payoff can be obtained with infinite memory:
The example of \figurename~\ref{fig:crazy} shows a one-player game. 
We claim that $(a)$ for $\MPinf$, player~$1$ can achieve
a threshold vector $(1,1)$, and $(b)$ for $\MPsup$,
player~$1$ can achieve a threshold vector $(2,2)$; 
$(c)$ if we restrict player~$1$ to use a finite-memory strategy, 
then it is not possible to win the multi-mean-payoff objective with threshold $(1,1)$
(and thus also not with $(2,2)$). 
To prove $(a)$, consider the strategy
that visits $n$ times $s_a$ and then $n$ times $s_b$, and repeats this forever
with increasing value of $n$. This guarantees a mean-payoff
vector $(1,1)$ for $\MPinf$ because in the long-run roughly half of the 
time is spent in $s_a$ and roughly half of the time in $s_b$.
To prove~$(b)$, consider the strategy that
alternates visits to $s_a$ and $s_b$ such that after the $n$th alternation,
the self-loop on the visited state $s$ ($s \in \{s_a,s_b\}$) is taken so
many times that the average frequency of $s$ gets larger than~$\frac{1}{n}$
in the current finite prefix of the play.
This is always possible and achieves threshold $(2,2)$ for $\MPsup$.
Note that the above two strategies require infinite memory. To prove $(c)$,
recall that finite-memory strategies produce an ultimately periodic play
and therefore $\MPinf$ and $\MPsup$ coincide. 
It is easy to see that such a play cannot achieve $(1,1)$ because the periodic
part would have to visit both $s_a$ and $s_b$ and then the mean-payoff vector $(v_1,v_2)$
of the play would be such that $v_1 + v_2 < 2$ and thus $v_1 = v_2 = 1$ is
impossible.
\hfill\qed
\end{proof}

\begin{figure}[!tb]
 \begin{center}

\begin{picture}(48,28)(0,0)
%\put(0,0){\framebox(45,28){}}

%\gasset{Nw=9,Nh=9,Nmr=4.5,rdist=1, loopdiam=6}

\node[Nmarks=i, iangle=180](n0)(10,12){$s_a$}
\node[Nmarks=n](n1)(40,12){$s_b$}
%\node[Nmarks=n](n2)(70,12){$q_2$}

\drawloop[ELside=l,loopCW=y, loopdiam=6](n0){$(2,0)$}
\drawloop[ELside=l,loopCW=y, loopdiam=6](n1){$(0,2)$}

%\drawloop[ELside=l,loopCW=y](nk){$0,1$}

\drawedge[ELpos=50, ELside=l, ELdist=0.5, curvedepth=6](n0,n1){$(0,0)$}
\drawedge[ELpos=50, ELside=l, curvedepth=6](n1,n0){$(0,0)$}

%\drawedge[dash={1}0](n3bis,nkbis){$0,1$}

\end{picture}
   \caption{A multi-mean-payoff game where infinite memory is necessary 
   to win (Lemma~\ref{lemm_inf_power}).}\label{fig:crazy}
 \end{center}
\end{figure}
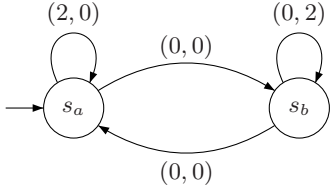

\noindent Lemma~\ref{thrm_inter} and Lemma~\ref{lemm_inf_power} %(proof in appendix), 
along with Theorem~\ref{thrm_gen_energy_fin} give the following result.

\begin{theorem}\label{thrm_gen_mean}
Multi-mean-payoff games are determined under finite-memory, but not determined by finite-memory 
(i.e., winning strategies in general require infinite-memory, and 
determinacy and determinacy under finite-memory do not coincide).
For multi-mean-payoff objectives $\varphi$ we have $\win{1}^{finite}\varphi=\win{1}^{fin-inf}\varphi$. 
\end{theorem}

\begin{comment}
{\bf KRISH TO LAURENT: ENUMERATION OF THINGS REMAINING.

\begin{enumerate}

\item Read the section below and check what is unclear; what we need to 
improve.

\item Definition section need to be modified in view of the section below.

\item What we write as Finite memory strategy section we should rename as 
Generalized Energy Games and add a paragraph that mean-payoff games with 
finite-memory strategies same as energy. The infinite-memory case in the 
following section.

\item The memoryless strategy case section we need to change the proof 
so that hardness is for two dimension and weights are $\set{-1,0,1}$.

\item Abstract, intro and conclusion (we can do this later).
\end{enumerate}
}
\end{comment}

\section{Multi-Mean-Payoff Games}\label{sec:multi-mean-payoff}
%\section{Generalized Mean-payoff Games with Infinite-memory Strategies}

In this section we consider multi-mean-payoff games with infinite-memory 
strategies (we have already shown in the previous section that multi-mean-payoff
games with finite-memory strategies coincide with multi-energy games).
We present the following complexity results for the mean-payoff 
threshold problem:
(1)~NP~$\cap$~coNP for conjunction of $\MeanPayoffSup$ objectives; 
(2)~coNP-completeness for conjunction of $\MeanPayoffInf$ objectives; and 
(3)~coNP-completeness for conjunction of mean-payoff-inf and mean-payoff-sup objectives.

\subsection{Conjunction of $\MeanPayoffSup$ objectives}\label{subsec:mean-sup}
We consider multi-weighted two-player game structures with the multi-mean-payoff 
objective $\MeanPayoffSup_G =\{ \pi \in \Plays(G) \mid \MPsup(\pi) \geq  (0,0,\ldots,0)\}$)
for player~$1$.
In general winning strategies for player~1 require infinite memory.
We show that memoryless winning strategies exist for player~2 and 
we present a reduction of the decision problem for a conjunction of $k$ mean-payoff-sup
objectives to solving polynomially many instances of the decision problem for single 
mean-payoff-sup objective.
As a consequence the decision problem for $\MeanPayoffSup_G$ lies in NP $\cap$ coNP, 
and we obtain a pseudo-polynomial time algorithm for this problem.

In the next lemma we show that if player~1 can satisfy the 
$\MeanPayoffSup$ objective in every individual dimension from all 
states, then player~1 can satisfy the conjunctive $\MeanPayoffSup$ 
objective from all states. The converse holds trivially.
The main idea of the proof is as follows: 
for each $1 \le i \le k$, let 
$\lambda_1^i$ be a winning strategy for player~1 for the objective 
$\MeanPayoffSup_i$.
Intuitively, the winning strategy for the conjunction of mean-payoff-sup objective plays 
$\lambda_1^i$, until the mean-payoff value on dimension $i$ gets larger than a 
number very close to $0$, and then switches to the strategy to $\lambda_1^{(i\pmod{k}) + 1}$, 
etc.
This way player~1 ensures nonnegative mean-payoff-sup value in every dimension.
We present the proof formally below.
While memoryless winning strategies exist for each individual dimension, 
we present a proof that does not use the assumption of witness memoryless 
winning strategies for individual dimensions. A similar proof technique is
used later where memoryless winning strategies for each individual dimension
are not guaranteed to exist.

\begin{lemma} \label{lem:sup}
If for all states $s \in S$ and for all $1 \le i \le k$, player~1 
has a winning strategy from $s$ for the objective 
$\MeanPayoffSup_i=\set{\pi \in \Plays \mid (\MPsup(\pi))_i \geq 0}$
(player~1 has winning strategies for each individual dimension),
then for all states $s\in S$, player~1 has a winning strategy 
from $s$ for the objective
$\MeanPayoffSup=\set{\pi \in \Plays \mid \MPsup(\pi) \geq (0,0,\ldots,0)}$.
\end{lemma}

\begin{proof}
%%[\ProofOfLem{lem:sup}]
For each $s\in S$ and $1 \leq i \leq k$, let $\lambda_1^i(s)$ be a winning strategy for
player~1 from $s$ for the objective $\MeanPayoffSup_i$, and consider the strategy tree $T_{\lambda_1^{i}(s)}$.
For $\alpha >0$, we say that a node $v$ of $T_{\lambda_1^{i}(s)}$ is an \emph{$\alpha$-good} 
node if the average of the weights of dimension~$i$ of the path from the root 
to $v$ is at least $-\alpha$.
For $Z \in \nat$, let $\widehat{T}_{\alpha}^{i,Z}(s)$ be the tree obtained from
$T_{\lambda_1^{i}(s)}$ by removing all descendants of the $\alpha$-good nodes 
that are at depth at least $Z$.
Hence, all branches of $\widehat{T}_{\alpha}^{i,Z}(s)$
have length at least $Z$, and the leaves are $\alpha$-good nodes.

We show that $\widehat{T}_{\alpha}^{i,Z}(s)$ is a finite tree.
%\begin{proposition}\label{prop:FiniteTree}
%The tree $\widehat{T}_{\lambda_{1,\alpha}^{i,j}(s)}$ is a finite tree.
%\end{proposition}
%
%\begin{proof}
%The proof is based on K\"{o}nig's Lemma. 
By K\"{o}nig's Lemma~\cite{Konig36}, it suffices to show that every path in the 
tree $\widehat{T}_{\alpha}^{i,Z}(s)$ is finite. 
Assume towards contradiction that there is an infinite path 
$\pi$ in $\widehat{T}_{\alpha}^{i,Z}(s)$. 
Then $\pi$ is a play consistent with $\lambda_1^{i}(s)$, and since $\pi$ does not 
contain any $\alpha$-good node beyond depth~$Z$, 
the mean-payoff-sup value of $\pi$ in dimension $i$ is at most $-\alpha$, 
i.e., $(\MPsup(\pi))_i \leq -\alpha$.
This contradicts the assumption that $\lambda_1^{i}(s)$ is a winning strategy 
for player~1 in dimension $i$.
%\hfill\qed
%\end{proof}

We now describe a strategy for player~1 based on the winning strategies of the 
individual dimensions and show that the strategy is winning for the conjunction
of mean-payoff-sup objectives. Let $W \in \nat$ be the largest absolute value 
of the weight function $w$.
\begin{algorithmic}[1]
\STATE $\alpha \gets 1$ \LOOP \FOR{$i=1$ to $k$} \STATE Let $s$ be
the current state, and $L$ be the length of the play so far.
 \STATE $Z\gets \frac{L \cdot W}{\alpha}$ %%\mynote{\mynote{L:updated the bound (was $j\gets 2^L$).}}
%\LOOP
 %\FOR{$m=1$ to $j$} 
\STATE Play according to $\lambda_1^{i}(s)$ until a leaf % $s'$ 
of $\widehat{T}_{\alpha}^{i,Z}(s)$ is reached.
% \STATE  $s\gets s'$
%\ENDFOR
  % \ENDLOOP
  \ENDFOR \STATE $\alpha \gets
\frac{\alpha}{2}$ \ENDLOOP
\end{algorithmic}

After the last command in the internal for-loop was executed, 
the mean-payoff value in dimension $i$, is at least 
$\frac{-L\cdot W - m \cdot \alpha}{L + m}$ where $m \geq \frac{L \cdot W}{\alpha}$
and this is at least $\frac{-L\cdot W - m \cdot \alpha}{m} \geq -2 \cdot \alpha$. 

Since $\widehat{T}_{\alpha}^{i,Z}(s)$ is a finite tree,
the main loop gets executed infinitely often (i.e., the strategy does not get stuck in the for-loop)
and $\alpha$ tends to $0$. Thus the supremum of the mean-payoff value is at least $0$ 
in every dimension. Hence the strategy described above is a winning strategy for player~1
for \MeanPayoffSup.
\hfill\qed
\end{proof}

In Lemma~\ref{lem:sup} the winning strategy constructed for player~1 requires 
infinite-memory, and by Lemma~\ref{lemm_inf_power} infinite memory is
required in general. For player~$2$, we show that memoryless winning strategies 
exist, and we derive the algorithmic solution for the mean-payoff threshold problem.

\begin{lemma}\label{lem:sup2}
%In generalized mean-payoff games with objective $\MeanPayoffSup = \{\pi \in \Plays_G \mid \MPsup%%%(\pi) \geq (0,0,\ldots,0)\}$,
%%for all states $s \in S$,
%if there is a winning strategy for player~2 from a state $s$ (to 
%falsify the objective $\MeanPayoffSup$), then there exists a \emph{memoryless} winning 
%strategy for player~2 form $s$.
In multi-mean-payoff games with conjunction of $\MeanPayoffSup$ objectives
for player~$1$, memoryless strategies are sufficient for player~$2$.
\end{lemma}

\begin{proof}
%%[\ProofOfLem{lem:GMPSup}]
The proof is by induction on the number of states $\abs{S}$ in the game structure.
The base case with $\abs{S} = 1$ is trivial.
We now consider the inductive case with $\abs{S} = n \geq 2$. 
Let $k \in \nat$ be the dimension of the weight function $w$. 
For $i = 1,\dots,k$, let $W_i = \win{2} \lnot \MeanPayoffSup_i$ be the winning 
region for player~2 for the one-dimensional mean-payoff game played in dimension $i$.
(i.e., in $W_i$ player~2 wins for the objective complementary to 
$\MeanPayoffSup_i = \{\pi \in \Plays \mid (\MPsup(\pi))_i \geq 0\}$).
Let $W=\bigcup_{i=1}^k W_i$. We consider the following two cases:

\begin{enumerate}

\item 
If $W=\emptyset$, then player~1 can satisfy the mean-payoff-sup objective
in every dimension, and then by Lemma~\ref{lem:sup} player~1 wins from 
everywhere for the objective $\MeanPayoffSup=\set{\pi \in \Plays \mid 
\MPsup(\pi) \geq (0,0,\ldots,0)}$. Hence there is no winning strategy 
for player~2.

\item If $W\neq\emptyset$, then there exists $1 \leq i \leq k$ such that $W_i \neq \emptyset$.
In $W_i$ there is a memoryless winning strategy $\lambda_2$ for player~2 
to falsify $\MeanPayoffSup_i = \{\pi \in \Plays \mid (\MPsup(\pi))_i \geq 0\}$
since memoryless winning strategies exist for both players in mean-payoff games 
with single objective~\cite{EM79}.
The strategy also falsifies $\MeanPayoffSup=\{\pi \in \Plays \mid \MPsup(\pi) \geq (0,0,\ldots,0)\}$.

Since $W_i$ is a winning region for player~$2$, it follows that $W_i=\AttrTwo{W_i}$,
and the graph $G'$ induced by $S \setminus W_i$ is a game structure.
%Let $G'$ be the game structure induced by $S \setminus W_i$, 
Let $W' = W \setminus W_i$ be the winning region for player~2 in $G'$. 
By induction hypothesis ($G'$ has strictly fewer states as a 
non-empty set $W_i$ is removed), it follows that there is a memoryless
winning strategy $\lambda_2'$ in $G'$ in the region $W'$.
The winning region $S \setminus (W_i \cup W')$ for player~1 in $G'$ 
is also winning for player~1 in $G$ (since $W_i=\AttrTwo{W_i}$, 
$G'$ is obtained by removing only player~1 edges).
Hence to complete the proof it suffices to show that the memoryless
strategy obtained by combining $\lambda_2$ in $W_i$ and $\lambda_2'$ 
in $W'$ is winning for player~2 from $W_i \cup W'$.
Define the strategy $\lambda_2^*$ as follows:
\[ 
\lambda^*_2(s) = \left\{ 
	\begin{array}{ll}
         \lambda_2(s) & \qquad \mbox{if $s\in W_i$}\\
         \lambda_2'(s) & \qquad \mbox{if $s\in W'$}.
	\end{array} \right. 
\]
Consider the memoryless strategy $\lambda_2^*$ for player~2 and 
the outcome of any counter strategy for player~1 that starts in $W'\cup W_i$.
There are two cases:
(a)~if the play reaches $W_i$, then it reaches in finitely many steps,
and then $\lambda_2$ ensures that player~2 wins; and 
(b)~if the play never reaches $W_i$, then the play always stays in $G'$, 
and now the strategy $\lambda_2'$ 
ensures winning for player~2.
This completes the proof of the second item.
\end{enumerate}
The desired result follows.
\hfill\qed
\end{proof}

%%{\bf KRISH: NEED TO INCLUDE YARON's EXAMPLE WHY CANNOT USE KOPCZYNSKI}

\paragraph{\bf Algorithm.} We present Algorithm~\ref{alg:mean-payoff-sup-solve} to solve 
games with conjunction of mean-payoff-sup objectives.
The algorithm maintains the current game structure $G_{\cur}$
induced by the current set of states $S_{\cur}$.
In every iteration of the repeat-loop, for $i=1,\dots,k$, 
we compute the winning region $W_i$ for player~$2$ in the current game structure 
with the single mean-payoff objective on dimension~$i$ by a call
to {\sf SolveSingleMeanPayoffSup}($G_{\cur}, (w)_i$) which returns the 
winning region for player~$1$ in $G_{\cur}$ for the objective $\MeanPayoffSup_i$.
If $W_i$ is nonempty, then we remove $W_i$ from the current game structure and 
the iteration continues.

% The formal description is as follows. 
% \begin{algorithmic}[1]
% %\LOOP 
% %\DO 
% %\LOOP 
% \FOR{$i=1$ to $k$} 
% \STATE Let $W_i\gets$ {\sc MPSingle}($G_{\mathit{cur}}, (w)_i$)
% \IF{$W_i\neq \emptyset$}
% \STATE $S_{\mathit{cur}} \gets S_{\mathit{cur}} \setminus W_i$  
% \STATE $G_{\mathit{cur}}\gets G_{\mathit{cur}} \upharpoonright S_{\mathit{cur}}$ 
% \STATE {\bf GOTO} Step~1
% \ENDIF
% \ENDFOR
% \STATE {\bf RETURN} $S_{\mathit{cur}}$
% %\ENDLOOP
% \end{algorithmic}

\begin{algorithm}[h]
\caption{{\sf SolveMeanPayoffSupGame}}
\label{alg:mean-payoff-sup-solve}
{
 \AlgData{A game $G$ with state space $S$ and multi-weight function $w$.}
 \AlgResult{The winning region of player~$1$ for objective $\MeanPayoffSup = \bigcap_{1 \leq i \leq k} \MeanPayoffSup_i$.}
 \Begin{
	\nl $G_{\cur} \gets G$ \;
	\nl $S_{\cur} \gets S$ \;
	\nl \Repeat{${LosingStatesFound} = false$}
	{
		\nl ${LosingStatesFound} \gets false$ \;
		\nl \For{$i=1$ to $k$}
		{
			\nl $W_i\gets S_{\cur} \setminus$ {\sf SolveSingleMeanPayoffSup}($G_{\cur}, (w)_i$) \quad {\tt /* solves $\MeanPayoffSup_i$ */}\;
			\nl \If{$W_i \neq \emptyset$}
			{
				\nl $S_{\cur} \gets S_{\cur} \setminus W_i$  \;
				\nl $G_{\cur} \gets G_{\cur} \upharpoonright S_{\cur}$ \;
				\nl ${LosingStatesFound} \gets true$ \;
			}			
		}
	}
	\nl \KwRet{$S_{\cur}$}\; 
 }
}
\end{algorithm}

In every iteration the set of states removed from the game structure
is certainly winning for player~2.
In the end we obtain a game structure such that player~1 wins the 
mean-payoff objective in every individual dimension from all states, 
and by Lemma~\ref{lem:sup} it follows that the remaining region is winning
for player~1. 
Thus game structures with conjunction of mean-payoff-sup objectives can 
be solved by $O(k \cdot \abs{S})$ calls to solutions of mean-payoff
games with single objective.
%%% L TO KRISH: removed the following sentence, and added a quick justification after Theorem 6 \ref{thrm:sup} %%%
%(this also gives a NP $\cap$ coNP upper bound for decision problem 
%whether player~1 has a winning strategy at a given state).
The following theorem summarizes the results for multi-weighted
games with conjunction of mean-payoff-sup objectives.

\begin{theorem}\label{thrm:sup}
For multi-weighted two-player game structures with objective 
$\MeanPayoffSup=\set{\pi \in \Plays \mid \MPsup(\pi) \geq (0,0,\ldots,0)}$
for player~1, the following assertions hold:
\begin{enumerate}
\item Winning strategies for player~1 require infinite-memory in general, 
and memoryless winning strategies exist for player~2.

\item The problem of deciding whether a given state is winning for player~1
lies in NP $\cap$ coNP.

\item The set of winning states for player~1 can be computed with  
$k\cdot \abs{S}$ calls to a procedure for solving game structures with 
single mean-payoff objective, hence in pseudo-polynomial time $O(k \cdot \abs{S}^2 \cdot \abs{E} \cdot W)$.

\end{enumerate}

\end{theorem}

%%\mynote{L: give one-two sentence to justify item $2$.}
The results of Theorem~\ref{thrm:sup} are proved as follows. 
Item~$1$ follows from Lemma~\ref{lemm_inf_power} and Lemma~\ref{lem:sup2}. 
Item~$3$ follows from Algorithm~\ref{alg:mean-payoff-sup-solve} and the 
results of~\cite{BCDGR09} where an algorithm  is given for games with 
single mean-payoff objectives that works in time $O(\abs{S} \cdot \abs{E} \cdot W)$.
We now present the details of Item~$2$ in two parts.
(1)~\emph{(In NP).} The NP algorithm guesses the winning region $W$ for 
player~1, and a memoryless winning strategy $\lambda_1^i$ for every 
individual dimension $i$ (such memoryless winning strategies for every 
individual dimension exist by the results of~\cite{EM79}). 
The verification procedure checks in polynomial time that for every dimension 
$i$ the set $W$ is the winning set for player~1 in the graph 
$G_{\lambda_1^i}$ using the polynomial time algorithm of~\cite{Karp}.
The correctness (that is, the existence of winning strategy in every individual dimension
implies winning for the conjunction) follows from Lemma~\ref{lem:sup}.
(2)~\emph{(In coNP).} The coNP algorithm guesses a memoryless winning 
strategy $\lambda_2$ for player~2. 
The verification procedure needs to solve mean-payoff-sup objectives 
for the graph $G_{\lambda_2}$ and by Algorithm~\ref{alg:mean-payoff-sup-solve}
this can be solved with $k \cdot \abs{S}$ calls to the polynomial time 
algorithm of~\cite{Karp} to solve graphs with single mean-payoff objectives.
Thus we have the polynomial-time verification procedure, and the coNP 
complexity bound follows.

% \begin{corollary}
% For multi-weighted two-player game structures with objective 
% $\MeanPayoffSup=\set{\pi \in \Plays \mid \MPsup(\pi) \geq (0,0,\ldots,0)}$
% for player~1, the set of winning states for player~1 can be computed in 
% time $O(k \cdot \abs{S}^2 \cdot \abs{E} \cdot W)$.
% \end{corollary}

\subsection{Conjunction of $\MeanPayoffInf$ objectives}\label{subsec:mean-inf}
We consider multi-weighted two-player game structures, and %a threshold vector$v$, and 
the multi-mean-payoff-inf objective 
$\MeanPayoffInf=\{ \pi \in \Plays(G) \mid \MPinf(\pi)\geq (0,0,\ldots,0)\}$)
for player~$1$.
In general winning strategies for player~1 require infinite memory (Lemma~\ref{lemm_inf_power}).
We show that memoryless winning strategies exist for player~2, and %%show 
the threshold problem is coNP-complete.

\paragraph{\bf Memoryless strategies for player~2.} The objective 
for player~2 is the complementary objective of player~1. 
It follows from the results of~\cite{Kop} that memoryless winning 
strategies exist for player~2 (see Appendix for discussion).
%The objective for player~1 is prefix-independent and convex, and then 
%it follows from the results of~\cite{Kop} on half-positional (memoryless) 
%games, that memoryless winning strategies exist for player~2 
%(the proof is based on induction on edges). 
%\mynote{L: Here, either say less, or say more (what is "prefix-independent", "convex", "half-positional").}

\paragraph{\bf Complexity.} We show that the problem of deciding whether 
a given state is winning for player~1 in multi-weighted game structures 
with conjunction of mean-payoff-inf objectives is coNP-complete. 
We first argue about the coNP lower bound.

\smallskip\noindent{\bf coNP lower bound.}
The proof is essentially the same as the proof of Lemma~\ref{thrm_hard} and relies
on the existence of memoryless winning strategies for player~2.
We consider the hardness proof of Lemma~\ref{thrm_hard} and the reduction 
used in the lemma.
If the formula is satisfiable, then consider the memoryless winning strategy 
for player~2 constructed from the satisfying assignment. 
Consider an arbitrary strategy (possibly with infinite-memory) for player~1.
%%Consider the literal $x$ that appear with highest frequency. 
Since the strategy for player~2 is constructed from a non-conflicting 
assignment, it follows that conflicting literals do not appear.
%%The highest frequency literal $x$ (the literal with highest mean-payoff-sup frequency) 
%%has frequency at least $\frac{1}{3\cdot n}$, where $n$ is the number of variables,
%%since one of the literal is visited every $3$ steps. 
%%The complement literal $\overline{x}$ is never visited, and every 
%%time literal $x$ is visited, the component corresponding to $\overline{x}$ 
%%decreases by~$1$.
Within every three steps some literal is visited. If $n$ is the number of 
variables, then in any play prefix compatible with the strategy of player~$2$,
the frequency of the literal~$x$ with highest frequency in this prefix is at least
$\frac{1}{3\cdot (n+1)}$ (and note that the literal $\overline{x}$ has never appeared). 
It follows that the average of the weights in the dimension for  
$\overline{x}$ is at most $-\frac{1}{3\cdot (n+1)}$ and therefore
the mean-payoff-inf objective is violated in some dimension.
Conversely, if the formula is not satisfiable, then against every memoryless 
strategy for player~2, the counter strategy constructed in Lemma~\ref{thrm_hard} 
(that alternates between the conflicting assignments) ensures that the 
mean-payoff-inf objective is satisfied. 
Hence the coNP-hardness follows.

\smallskip\noindent{\bf coNP upper bound.}
The rest of the section is devoted to proving the coNP upper bound. 
Once a memoryless strategy for player~2 is fixed (as a polynomial witness),
we obtain a one-player game structure. 
To establish the coNP upper bound we need to show that the problem can 
be solved in polynomial time for one-player game structures.
A polynomial-time algorithm for the problem is obtained by solving a
variant of the zero circuit problem for multi-weighted directed graphs.
The variant of the zero circuit problem is the \emph{nonnegative multi-cycle} 
problem for directed graphs, where the multi-cycle is not required to be
connected by edges as in the case of zero circuit problem.

\paragraph{\bf Nonnegative multi-cycles.} Let $G=(V,E,w:E \to \integ^k)$ be a 
multi-weighted directed graph that is strongly connected.
A \emph{multi-cycle} is a multi-set of simple cycles.
For a multi-cycle $\mathbf{C}$ we denote by $\SetCycle(\mathbf{C})$ the 
set of cycles that appear in $\mathbf{C}$, and hence $\SetCycle(\mathbf{C})$ 
is a set of simple cycles.
For multi-cycle $\mathbf{C} = \{C_1,\dots,C_n\}$ we denote with $m_i$ the number 
of occurrences of a simple cycle $C_i$ in the multi-set $\mathbf{C}$, and refer 
to $m_i$ as the \emph{factor} of $C_i$.
For a simple cycle $C = (e_0,e_1 \dots e_n)$, we denote $w(C) = \sum_{e\in C} w(e)$.
For a multi-cycle $\mathbf{C}$, we denote $w(\mathbf{C})=\sum_{C \in \mathbf{C}} w(C)$ (note that
in the summation a cycle $C$ may appear multiple times in $\mathbf{C}$, and alternatively
the summation can be expressed as considering simple cycles $C_i$ that appear in $\mathbf{C}$ 
and summing up $m_i \cdot w(C_i)$).
A \emph{nonnegative multi-cycle} is a non-empty  multi-set of simple cycles 
$\mathbf{C}$ such that $w(\mathbf{C}) \geq 0$ %(resp. $w(\mathbf{C}) \geq 0$) 
(i.e., in every dimension the weight is nonnegative).

\begin{lemma} \label{lem:multiCycleProps}
Let $G=(V,E,w:E\to\integ^k)$ be a multi-weighted directed graph that is strongly connected.
\begin{enumerate}
\item The problem of deciding if $G$ has a nonnegative multi-cycle can be solved in polynomial time.
\item If $G$ does not have a nonnegative multi-cycle, 
then there exist a constant $m_G\in\Nat$ and a real-valued constant $c_G > 0$ such that 
for all finite paths $\pi^f$ in the graph $G$ we have 
$\min\{w_i(\pi^f)\mid i\in\{1,\dots,k\}\} \leq m_G -c_G \cdot \abs{\pi^f}$. 

%\begin{enumerate}
%    \item If $G$ does not have a positive multi-cycle then $w_i(\pi) \leq m_G$.
%    \item If $G$ does not have a nonnegative multi-cycle then $w_i(\pi) \leq m_G -c_G \abs{\pi}$.
%\end{enumerate}
\end{enumerate}
\end{lemma}

\begin{proof}
We prove both the items below.
\begin{enumerate}
\item 
The proof of the first item is almost exactly as the proof of Theorem 2.2 in~\cite{KS88}.
Given the directed strongly connected graph $G=(V,E,w:E\to\integ^k)$, we consider 
a variable $x_e$ (for edge coefficient of $e$) for every $e\in E$. 
We define the following set of linear constraints.
\begin{enumerate}
\item For $v\in V$, let $\mathit{IN}(v)$ be the set of all in-edges of $v$, and $\mathit{OUT}(v)$ be the 
set of out-edges of $v$. For every $v\in V$ we define the linear constraint that $\sum_{e\in\mathit{IN}(v)} x_e = \sum_{e\in\mathit{OUT}(v)} x_e$.
\item For every $e\in E$ we define the constraint $x_e\geq 0$.
\item For every dimension $i\in \set{1,\dots,k}$, we define the constraint $\sum_{e\in E} x_e \cdot w_i(e) \geq 0$. %%(resp $\sum_{e\in E} x_e w_i(e) \geq 0$).
\item Finally, we define the constraint $\sum_{e\in E} x_e \geq 1$.
\end{enumerate}
The first set of linear constraints is intuitively the flow constraints; 
the second constraint specifies that for every edge $e$, the edge coefficient  $x_e$ 
is nonnegative;
the third constraint specifies that in every dimension the sum of edge 
coefficient time the weights is nonnegative; and 
the last constraint ensures that at least one edge coefficient is strictly
positive (to ensure that the multi cycle is non-empty). 
This set of constraints can be solved in polynomial time using standard 
linear programming algorithms. 
It essentially follows from~\cite{KS88} that this set of linear constraints has a 
solution iff a nonnegative multi-cycle exists.

%The proof of Proposition \ref{prop:multiCycleProps}(2) is trivial if $G$ does not have a positive multi-cycle, but it requires some technical work for the case where $G$ does not have a nonnegative multi-cycle.
\item 
Let $\pi^f$ be a finite path in $G$. 
The finite path $\pi^f$ can be decomposed into three paths 
$\pi_0^f,\pi^f_c,\pi_1^f$ where $\pi_0^f$ is an initial prefix of length 
at most $\abs{V}$, $\pi_c^f$ consists of cycles (not necessarily simple), 
and $\pi_1^f$ is a segment of length at most $\abs{V}$ in the end. 
We can uniquely decompose $\pi_c$ into a set $\mathbf{C}$ of multi-cycles and 
hence also into a set of \emph{simple} cycles 
$\wh{C}=\SetCycle(\mathbf{C}) = \set{C_1, \dots, C_n}$, for $n \leq 2^{\abs{E}}$, such that cycle 
$C_i$ occurs $r_i$ times in $\pi_c$, for some $r_i \in \Nat$.
The sum of the weights in the part of $\pi_c^f$ is 
\[
w(\pi_c^f)= \sum_{i=1}^n r_i \cdot w(C_i) = 
(\sum_{i=1}^n r_i) \cdot \sum_{i=1}^n \frac{r_i}{(\sum_{i=1}^n r_i)} \cdot w(C_i)
\leq 
\abs{\pi_c^f} \cdot \sum_{i=1}^n \frac{r_i}{ (\sum_{i=1}^n r_i)} \cdot w(C_i).
\]
The second equality is obtained by multiplying and dividing with $(\sum_{i=1}^n r_i)$,
and the inequality is obtained since 
$(\sum_{i=1}^n r_i) \leq \abs{\pi_c^f}$ (as $\abs{\pi_c^f}= \sum_{i=1}^n r_i \cdot \abs{C_i}$).
Let $\beta_i=\frac{r_i}{(\sum_{i=1}^n r_i)}$ and observe that
$\beta_1, \dots, \beta_n \geq 0$ with $\sum_{j=1}^n \beta_j = 1$.
We first show the existence of a constant $\eta_{\wh{C}}>0$, 
such that for every $\alpha_1, \dots, \alpha_n \geq 0$ with
$\sum_{j=1}^n \alpha_j = 1$, there exists a dimension $i\in \set{1,\dots,k}$ 
such that $\sum_{j=1}^n \alpha_j \cdot w_i(C_j) \leq -\eta_{\wh{C}}$.

For every $i\in\set{1,\dots, k}$, we define a function 
$f_i(\alpha_1,\dots,\alpha_n) = \sum_{j=1}^n \alpha_j \cdot w_i(C_j)$ and 
$f(\alpha_1,\dots,\alpha_n) = \min\set{f_i(\alpha_1,\dots,\alpha_n) \mid 
1\leq i \leq k}$.
For every $i \in \set{1,\ldots,k}$, the function $f_i$ is continuous.
Since $f$ is the minimum of a finite number of continuous functions, 
$f$ is also continuous.
Observe that $[0,1]^n \cap \set{(\alpha_1,\dots, \alpha_n) \mid 
\sum_{j=1}^n \alpha_j = 1}$ is a closed and bounded set.
Hence by \emph{Weierstrass theorem} the function $f$ has a maxima $c_f$ in this domain.
Let $\alpha ^*_1, \dots, \alpha ^*_n \geq 0$ such that $f(\alpha ^*_1, \dots, \alpha ^*_n) = c_f$ and $\sum_{j=1}^n \alpha^*_j = 1$.
Assume towards contradiction that $c_f \geq 0$, we then show that the linear programming problem on 
the constraints mentioned above (in item~1) has a solution, which leads to a contradiction.
For an edge $e$, we define the edge coefficient as follows: 
$x_e=\sum_{e \in C_j \in \wh{C}} \alpha_j^*$ (i.e., the sum of the $\alpha^*_j$'s of the
cycle the edge belongs to). 
It follows that all the constraints are satisfied, and this contradicts the assumption 
that there is no nonnegative multi-cycle.
Hence we have $c_f < 0$.
Hence it follows that there exists a dimension $i$ such that 
\[
\begin{array}{rcl}
w_i(\pi^f) & \leq &   
(\abs{\pi_0^f} + \abs{\pi_1^f}) \cdot W + c_f \cdot \abs{\pi^f_c} 
=
(\abs{\pi_0^f} + \abs{\pi_1^f}) \cdot W + 
(\abs{\pi_0^f} + \abs{\pi_1^f}) \cdot (-c_f) + 
c_f \cdot \abs{\pi^f} \\
& \leq & 
2\cdot \abs{V} \cdot (W-c_f) + c_f \cdot \abs{\pi^f}
.
\end{array}
\]
Let $m_{\wh{C}}= \lceil 2\cdot \abs{V} \cdot (W-c_f) \rceil$ and 
$\eta_{\wh{C}}=-c_f$, and we obtain the desired result for the
path $\pi^f$.
Let $\mathcal{C}=\set{\SetCycle(\mathbf{C}) \mid \mathbf{C} \text{ is a multi-cycle}}$ 
be the set of simple cycles of all the multi-cycles of $G$.
Note that $\mathcal{C}$ is a set whose elements are subsets of simple cycles, 
i.e., $\mathcal{C}$ is the power set of power set of simple cycles and hence 
$\abs{\mathcal{C}} \leq 2^{2^{\abs{E}}}$.
By choosing $m_G=\max_{\wh{C} \in \mathcal{C}} m_{\wh{C}}$ and 
$c_G=\min_{\wh{C} \in \mathcal{C}} \eta_{\wh{C}}$
we obtain the desired result. %%\mynote{L: is $\mathcal{C}$ a finite set ?}
\end{enumerate}
%Thus we have the desired result.
\hfill\qed
\end{proof}

\noindent In sequel we abbreviate a maximal strongly connected component of a graph 
as a scc.

\begin{lemma}\label{lemm_inf_1}
Let $G$ be a multi-weighted one-player game structure, and let $s_0$ be the
initial state. 
If there is a scc $C$ reachable from $s_0$ such that the multi-weighted directed 
graph induced by $C$ has a nonnegative multi-cycle, then player~1 has a 
strategy to satisfy the mean-payoff-inf objective $\MeanPayoffInf$. 
\end{lemma}

\begin{proof}
Let $C$ be a scc reachable from $s_0$ 
%%(within $\abs{S}$ steps \mynote{L: is it useful to mention this bound here ?}) 
such that the graph induced by $C$ has a nonnegative multi-cycle.
Then there exist simple cycles $C_1, \dots, C_n$, factors $m_1, \dots, m_n$ and 
finite paths $\pi_{1,2}, \pi_{2,3}, \dots, \pi_{n-1,n}, \pi_{n,1}$ such that
(i)~the path $\pi_{i,j}$ is an acyclic path from $C_i$ to $C_j$, and  
(ii)~for every $i = 1, \dots, k$, we have $\sum_{j=1}^n m_j \cdot w_i(C_j) 
\geq 0$.
%\begin{enumerate}
%    \item the path $\pi_{i,j}$ is a path from $C_i$ to $C_j$ of length at most $\abs{S}$. \mynote{L: is it useful to mention this bound here ?}
%
%    \item For every $i = 1, \dots, k$, we have $\sum_{j=1}^n m_j \cdot w_i(C_j) \geq 0$
%\end{enumerate}
An infinite memory strategy for player~1 is as follows: initialize $Z=1$, and 
follow the steps below: 
\begin{algorithmic}[1]
\LOOP
\STATE $Z \cdot m_1$ times in cycle $C_1$
\STATE $\pi_{1,2}$
\STATE $Z \cdot m_2$ times in cycle $C_2$
\STATE $\pi_{2,3}$
\STATE $\cdots$
\STATE $Z \cdot m_n$ times in cycle $C_n$
\STATE $\pi_{n,1}$
\STATE $Z \gets Z+1$
\ENDLOOP
\end{algorithmic}
Let $L=\abs{\pi_{1,2}} + \abs{\pi_{2,3}} +\cdots \abs{\pi_{n-1,n}} +\abs{\pi_{n,1}}$ be the sum 
of the lengths of the paths between cycles, and 
let $P=\abs{C_1} + \abs{C_2} +\cdots + \abs{C_n}$ be the sum of the lengths of the cycles.
Note that both $L$ and $P$ are bounded by $2^{\abs{E}} \cdot \abs{S}$
as $n \leq 2^{\abs{E}}$ 
and each path and cycle is of length at most $\abs{S}$. 
%%\mynote{L: is it useful to mention this bound here ?}
Consider the steps executed in round $Z+1$: the sum of weights due to executing the
cycles in all previous rounds up to $Z$ is nonnegative in all dimensions.
Hence the sum of weights in any dimension, in the steps executed in round $Z+1$ 
is at least 
\[
-( \abs{S}+ (Z+1)\cdot P + Z \cdot L + L) \cdot W.
\]
The negative contribution can come from executing the initial prefix of length at most
$\abs{S}$ to reach the scc, then the cycles in the present round 
(bounded by $(Z+1)\cdot P$ steps) and the paths $\pi_{i,j}$ of length at most~$L$ in
the previous $Z$ rounds and in the current round (in total bounded by 
$Z \cdot L +L$ steps).
The number of steps executed so far is at least $(L+P) \cdot \sum_{i=1}^Z i =
(L+P)\cdot \frac{Z\cdot(Z+1)}{2} \geq \frac{(L+P)\cdot Z^2}{2}$. 
%%for sufficiently large $Z$.
Hence the average for all dimensions for all steps in round $Z+1$ is at least
\[
\begin{array}{rcl}
\displaystyle
\frac{-2\cdot (\abs{S} + (Z+1)\cdot P + Z \cdot L + L) \cdot W}{(L+P)\cdot Z^2}
& = & 
\displaystyle
\frac{-2\cdot (\abs{S} + (Z+1)\cdot (P+L) ) \cdot W}{(L+P)\cdot Z^2}
\\[2ex]
& \geq & 
\displaystyle
\frac{-2\cdot \abs{S}\cdot W}{Z^2} + \frac{-2\cdot (Z+1)\cdot W}{Z^2}.
\end{array}
\]
%%\mynote{L: not sure how the previous equality is obtained.}
As $Z \to \infty$, it follows that the mean-payoff-inf value is at least~0 
in every dimension, and hence the result follows.
\hfill\qed
\end{proof}

\begin{lemma}\label{lemm_inf_2}
Let $G$ be a multi-weighted one-player game structure, and let $s_0$ be the
initial state. 
If for every scc $C$ reachable from $s_0$ the multi-weighted directed 
graph induced by $C$ does not have a nonnegative multi-cycle, then player~1 does not 
have strategy from $s_0$ to satisfy the mean-payoff-inf objective $\MeanPayoffInf$. 
\end{lemma}

\begin{proof}
Consider an arbitrary strategy for player~1, and let the set of states visited 
infinitely often be contained in an scc $C$.
Since $C$ does not have a nonnegative multi-cycle it follows from 
Lemma~\ref{lem:multiCycleProps}(2) that every infinite path that visits states in
$C$ has a mean-payoff-inf value at most $-c$, for some $c > 0$, in some dimension.
It follows the strategy for player~1 does not satisfy the mean-payoff-inf objective 
$\MeanPayoffInf$.
\hfill\qed
\end{proof}

The following lemma shows that in one-player game structure the $\MeanPayoffInf$ 
objective can be solved in polynomial time.
To describe the precise complexity, let us denote by $\mathsf{LP}(i,j)$ 
the complexity to solve linear inequalities with $i$ variables and $j$ constraints.

\begin{lemma}
Given a multi-weighted one-player game structure $G$ and a state $s_0$, 
the problem of deciding whether player~1 has a strategy for a mean-payoff-inf 
objective $\MeanPayoffInf$ from $s_0$ can be solved in polynomial time 
(in time $O(|S|+|E|) + \mathsf{LP}(|E|, |S|+|E|+k+1)$). 
\end{lemma}

\begin{proof}
It follows from Lemma~\ref{lemm_inf_1} and Lemma~\ref{lemm_inf_2} that an algorithm to 
solve the problem is as follows: consider the scc decomposition of the graph,
and for every multi-weighted graph induced by an scc $C$ reachable from $s_0$ 
check if the multi-weighted directed graph induced by $C$ has a nonnegative 
multi-cycle (in polynomial time by Lemma~\ref{lem:multiCycleProps}(1)). 
Since scc decomposition is linear time (in time $O(|S|+|E|)$) and the number of 
scc's is linear,  we obtain the desired result.
The complexity of the linear inequations follows from Lemma~\ref{lem:multiCycleProps}. 
\hfill\qed
\end{proof}

Thus we obtain the desired coNP upper bound.
We have the following theorem summarizing the result of this section.

\begin{theorem}\label{thrm:inf}
For multi-weighted two-player game structures with objective 
$\MeanPayoffInf=\set{\pi \in \Plays \mid \MPinf(\pi) \geq (0,0,\ldots,0)}$
for player~1, the following assertions hold:
\begin{enumerate}
\item Winning strategies for player~1 require infinite-memory in general, 
and memoryless winning strategies exist for player~2.

\item The problem of deciding whether a given state is winning for player~1
is coNP-complete.

%%\mynote{L: give complexity for one-player game (in terms of $LP(\cdot)$).}

\end{enumerate}

\end{theorem}

\subsection{Conjunction of $\MeanPayoffInf$ and $\MeanPayoffSup$ objectives}
We consider multi-weighted two-player game structures, 
two sets $I,J \subseteq \{1,\dots,k\}$, and the multi-mean-payoff objective 
$\MeanPayoffInfSup(I,J)= \{ \pi \in \Plays(G) \mid 
\forall i \in I: \MPinf(\pi)_i \geq 0 \text{ and } 
\forall j \in J: \MPsup(\pi)_j \geq 0 \}$ for player~1.
%, i.e., we require the mean-payoff-inf objective in the first 
%$j$ dimensions, and the mean-payoff-sup objective is the rest of
%the dimensions.

Note that the problem is more general than the problem considered in the 
previous section (with $J=\emptyset$ we obtain $\MeanPayoffInf$ objectives,
and with $I=\emptyset$ we obtain $\MeanPayoffSup$ objectives).
Hence it follows that in general winning strategies for player~1 
require infinite-memory, and the problem is coNP-hard.
We show that memoryless winning strategies exist for player~2, 
and that the decision problem is coNP-complete.

We start with the crucial result that considers the case when the 
mean-payoff-sup objective is required for one dimension, and for all the 
other dimensions the mean-payoff-inf objective is required.
The lemma shows that if only one dimension is $\MeanPayoffSup$ objective,
then it can be equivalently considered as $\MeanPayoffInf$ objective.

%%\mynote{L: why next lemma is formulated for player~$2$ ?}
\begin{lemma} \label{lem:GMPmixedOneSup}
Let $I=\{1,\dots,k-1\}$ and $s$ be a state.
Player~1 has a winning strategy for the objective  
$\MeanPayoffInfSup(I,\{k\})$ from $s$ if and only if player~1 has a winning strategy 
for the objective $\MeanPayoffInf = \MeanPayoffInfSup(I\cup\{k\},\emptyset)$ from $s$.
%%Let $S = \singelton{1,\dots, k-1}$, then player-1 is the winner for the $\ConjMPInfSupGeq{0}{S}$ objective $\Leftrightarrow$ player-1 is the winner for the $\ConjMPInfGeq{0}$ objective (on $k$ dimensions).
\end{lemma}

\begin{proof}%%[\ProofOfProp{prop:GMPmixedOneSup}]
To prove the lemma we show the following equivalent statement:
Player~2 has a winning strategy to falsify 
$\MeanPayoffInfSup(I,\{k\})$ from $s$ if and only if player~2 has a winning strategy 
to falsify $\MeanPayoffInf = \MeanPayoffInfSup(I\cup\{k\},\emptyset)$ from $s$.

One direction is trivial as for any sequence $(u_i)_{i\geq 0}$ of real numbers we 
have $\limsup_{i\to\infty} u_i \geq \liminf_{i\to\infty} u_i$, and hence
it follows that a winning strategy for player~2 to falsify $\MeanPayoffInfSup(I,\{i\})$
is also a winning strategy to falsify $\MeanPayoffInf$.

Suppose that player~2 has a winning strategy for $\MeanPayoffInf$, then by 
Theorem~\ref{thrm:inf} player~2 has a memoryless winning strategy~$\lambda_2$.
Let $G_{\lambda_2}$ be the one-player game structure obtained by fixing the 
strategy $\lambda_2$ for player~2.
Since $\lambda_2$ is winning for player~$2$, it follows from Lemma~\ref{lemm_inf_1}
that in $G_{\lambda_2}$, for all scc's $C$, in the subgraph induced by $C$ there 
is no nonnegative multi-cycle.
It follows from Lemma~\ref{lem:multiCycleProps} that
there exist a constant $m_{G_{\lambda_2}} \in\Nat$ and a real-valued constant 
$c_{G_{\lambda_2}} > 0$ such that for all finite paths $\pi^f$ in the graph $G$ we have 
$\min\{w_i(\pi^f)\mid i\in\{1,\dots,k\}\} \leq m_{G_{\lambda_2}} -c_{G_{\lambda_2}} \cdot \abs{\pi^f}$.
Let us denote $c = c_{G_{\lambda_2}}$. %%be the constant from Lemma~\ref{lem:multiCycleProps}.
We show that $\lambda_2$ is winning for player~$2$ (to falsify $\MeanPayoffInfSup(I,\{k\})$).
Consider a play $\pi$ consistent with $\lambda_2$, and assume that $\MPsup(\pi)_k \geq 0$.
Then the average payoff in dimension $k$ is greater than $-\frac{c}{2}$ in 
infinitely many positions (since the limit-superior is at least~0), and 
by Lemma~\ref{lem:multiCycleProps} there is a dimension $1 \leq i < k$
with average payoff at most $-c$ in infinitely many positions, thus $\MPinf(\pi)_i < 0$.
Hence either the supremum of the average weight in dimension $k$ is negative, 
or the  infimum of the average weight in one of the other dimensions is negative.
In either case, the strategy $\lambda_2$ is winning for player~2. 
This completes the proof.
\hfill\qed
\end{proof}

Our goal is now to prove a result similar to Lemma~\ref{lem:sup} for 
$\MeanPayoffInfSup(I,J)$ objectives. 
To prove the result, we first prove two lemmas.
The following lemma about $\MeanPayoffInf$ objectives 
is derived from the proof of Lemma~\ref{lemm_inf_1} and
it shows that if player~1 has a winning strategy for a mean-payoff-inf objective
(with threshold $0$ in every dimension), then for every $\alpha>0$ there is a 
finite-memory strategy to ensure mean-payoff-inf value of at least $-\alpha$ in every dimension.
Lemma~\ref{lemm:mixed_req_2} will be a consequence of Lemma~\ref{lemm:mixed_req_1}.

\begin{lemma}\label{lemm:mixed_req_1}
Let $G$ be a multi-weighted two-player game structure, and let $s_0$ be the
initial state. 
If there is a winning strategy for player~1 for the objective 
$\MeanPayoffInf=\set{ \pi \in \Plays(G) \mid 
\forall 1 \leq i \le k. \ (\MPinf(\pi))_i\geq 0}$,
then for all $\alpha>0$ there is a finite-memory winning strategy for 
player~1 to ensure the objective
$\MeanPayoffInf(-\alpha)=\set{ \pi \in \Plays(G) \mid 
\forall 1 \leq i \le k. \ (\MPinf(\pi))_i\geq -\alpha}$.
\end{lemma}

\begin{proof}
Since against finite-memory strategies for player~1 memoryless 
winning strategies exist for player~2 (Lemma~\ref{thrm_inter} and 
Lemma~\ref{lem:player-two-memoryless})
and multi-mean-payoff games are determined under finite memory (Theorem~\ref{thrm_gen_mean}) 
to prove that finite-memory winning strategies exist for player~1 
for the objective $\MeanPayoffInf(-\alpha)$ we show that against 
every memoryless strategy for player~2 there exists a  
finite-memory winning strategy for player~1.
Consider a memoryless strategy for player~2 and the one-player 
game structure obtained after fixing the strategy.
By Lemma~\ref{lemm_inf_2}, since player~1 satisfies the $\MeanPayoffInf$ objective, 
there must be a scc $C$ reachable from $s_0$ (within $\abs{S}$ steps) such that 
the graph induced by $C$ has a nonnegative multi-cycle.
Then there exist simple cycles $C_1, \dots, C_n$, factors $m_1, \dots, m_n$ and 
finite paths $\pi_{1,2}, \pi_{2,3}, \dots, \pi_{n-1,n}, \pi_{n,1}$ such that:
\begin{enumerate}
    \item the path $\pi_{i,j}$ is a path between $C_i$ to $C_j$ with length at most $\abs{S}$.
    \item For every $i = 1, \dots, k$, we have $\sum_{j=1}^n m_j \cdot w_i(C_j) \geq 0$
\end{enumerate}
A finite memory strategy for player~1 is as follows: for large enough $Z$, 
follow the steps below: 

\begin{algorithmic}[1]
\LOOP
\STATE $Z \cdot m_1$ times in cycle $C_1$
\STATE $\pi_{1,2}$
\STATE $Z \cdot m_2$ times in cycle $C_2$
\STATE $\pi_{2,3}$
\STATE $\cdots$
\STATE $Z \cdot m_n$ times in cycle $C_n$
\STATE $\pi_{n,1}$
\ENDLOOP
\end{algorithmic}

In contrast with the strategy of Lemma~\ref{lemm_inf_1}, the above strategy 
plays the same in every round but for large enough $Z$, thus it can be implemented with finite memory.
Let $L=\abs{\pi_{1,2}} + \abs{\pi_{2,3}} +\cdots \abs{\pi_{n-1,n}} +\abs{\pi_{n,1}}$ be the sum 
of the lengths of the paths between cycles, and 
let $M=\abs{C_1} + \abs{C_2} +\cdots + \abs{C_n}$ be the sum of the lengths of the cycles.
Note that both $L$ and $M$ are bounded by $2^{\abs{E}} \cdot \abs{S}$ as $n \leq 2^{\abs{E}}$ 
and each path and cycle is of length at most $\abs{S}$.
Consider the steps executed in round $i$: the sum of weights due to executing the
cycles in all previous rounds up to $Z$ is nonnegative in all dimensions.
Hence the sum of weights in any dimension, in the steps executed in round~$i$ 
is at least 
\[
-( \abs{S}+ Z\cdot M + i \cdot L + L) \cdot W.
\]
The argument is as in Lemma~\ref{lemm_inf_1}.
The number of steps executed so far is at least 
$(L+M) \cdot (i-1) \cdot Z$.
Hence the average for all dimensions for all steps in round $i$ is at least
\[
-\frac{( \abs{S}+ (i+1)\cdot L + Z\cdot M) \cdot W}{(L+M) \cdot (i-1) \cdot Z}
\geq 
-\left(\frac{\abs{S}\cdot W}{Z} + \frac{2\cdot W}{Z} + \frac{W}{(i-1)} \right),
\]
for $i \geq 3$.
With $Z$ large enough ($Z \geq \frac{(\abs{S} + 2)\cdot W}{\alpha}$), it follows 
that as $i \to \infty$, the mean-payoff-inf value is at least~$-\alpha$ 
in every dimension, and hence the result follows.
\hfill\qed
\end{proof}

\begin{lemma}\label{lemm:mixed_req_2}
Let $G$ be a multi-weighted two-player game structure, and let $s_0$ be the
initial state. 
If there is a winning strategy for player~1 for the objective 
$\MeanPayoffInf=\set{ \pi \in \Plays(G) \mid 
\forall 1 \leq i \le k. \ (\MPinf(\pi))_i\geq 0}$,
then for all $\alpha>0$ there is a finite-memory winning strategy $\lambda$ 
and a number $N_{\alpha,\lambda,s_0}$ such that against all strategies of 
player~2 and for all $n \in \Nat$ the sum of weights after $n$ steps
is at least $-(N_{\alpha,\lambda,s_0} +n)\cdot \alpha$ in every dimension,
i.e., the average of the weights is at least $-2\cdot \alpha$ once $n \geq N_{\alpha,\lambda,s_0}$.
%%\mynote{L: why factor $\alpha$ in $N_{\alpha,\lambda,s_0} \cdot \alpha$ ?}
\end{lemma}

\begin{proof}
Fix a finite-memory strategy $\lambda$ for player~1 to satisfy the objective
$\MeanPayoffInf(-\alpha)=\set{ \pi \in \Plays(G) \mid 
\forall 1 \leq i \le k. \ (\MPinf(\pi))_i\geq -\alpha}$
(such a strategy exists by Lemma~\ref{lemm:mixed_req_1}).
Let $M$ be the size of the memory. 
%%\mynote{L: notation $M$ may be confusing with the bound $M$ in Lemma~\ref{lemm:mixed_req_1}.}
In the game structure obtained by fixing the strategy, 
in all cycles the average of the weights in every dimension is at least $-\alpha$.
For any path it can be decomposed into initial prefix and 
a cycle free segment in the end (each of length at most $M \cdot \abs{S}$), and 
the other part is decomposed into cycles (not necessarily simple cycles)
(as done in Lemma~\ref{lem:multiCycleProps}).
The initial prefix and trailing prefix is of length at most $M \cdot \abs{S}$ and 
the sum of the weights is at least $-2\cdot M \cdot \abs{S} \cdot W$.
Hence choosing $N_{\alpha,\lambda,s_0} \geq \frac{2\cdot M \cdot \abs{S} \cdot W}{\alpha}$
proves the desired result.
\hfill\qed
\end{proof}

\begin{lemma} \label{lemm:mixed_1}
Let $G$ be a multi-weighted game structure with 
multi-mean-payoff objective 
$\MeanPayoffInfSup(I,J)= \{ \pi \in \Plays(G) \mid 
\forall i \in I: \MPinf(\pi)_i \geq 0 \text{ and } 
\forall j \in J: \MPsup(\pi)_j \geq 0 \}$ for player~1.
For $\ell \in J$, let 
$\Phi_{\ell}= \MeanPayoffInfSup(I,\{\ell\})$
denote the objective that requires to satisfy all 
$\MeanPayoffInf$ objectives and the $\MeanPayoffSup$
objective in dimension $\ell$.
If for all states $s \in S$ and for all $\ell \in J$, player~1 
has a winning strategy from $s$ for the objective $\Phi_{\ell}$,
then for all states $s\in S$, player~1 has a winning strategy 
from $s$ for the objective $\MeanPayoffInfSup(I,J)$.
\end{lemma}

The key idea of the proof is similar to Lemma~\ref{lem:sup} and we use
Lemma~\ref{lemm:mixed_req_1} 
(details are presented below for completeness).
For all $s\in S$ and all $\ell \in J$, let $\lambda_1^\ell(s)$ be a winning 
strategy from $s$ for player~1 for the objective $\Phi_{\ell}$.
Intuitively, the winning strategy for the conjunction of mean-payoff objectives 
plays $\lambda_1^\ell(\cdot)$ until the mean-payoff value in dimension $\ell$ 
gets very close to $0$, and then switches to a strategy for another value of $\ell \in J$.
Thus player~1 ensures nonnegative mean-payoff value in every dimension,
with mean-payoff-inf in dimensions of $I$ and mean-payoff-sup in dimensions of $J$.

%\mynote{L: todo: update next proof as in proof of Lemma~\ref{lem:sup}.}
\begin{proof} 
%We present the formal proof below.
Let $\alpha>0$, and $s$ be the initial state.
Let $\Phi_{\ell}(-\frac{\alpha}{2})= \{ \pi \in \Plays(G) \mid 
\forall i \in I: (\MPinf(\pi))_i \geq -\frac{\alpha}{2} \text{ and } 
(\MPsup(\pi))_{\ell} \geq -\frac{\alpha}{2} \}$.
Let $\lambda_{1,\alpha}^\ell(s)$ be a finite-memory 
winning strategy for player~1 for the objective 
$\Phi_{\ell}(-\frac{\alpha}{2})$ with the initial 
state $s$ (the existence of finite-memory winning strategy for 
$\Phi_{\ell}(-\frac{\alpha}{2})$ follows from Lemma~\ref{lem:GMPmixedOneSup} and Lemma~\ref{lemm:mixed_req_1}).
For $Z \in \nat$, consider the tree 
$\widehat{T}_{\lambda_{1,\alpha}^{\ell,Z}(s)}$ defined as follows.
Let $T_{\lambda_{1,\alpha}^{\ell}(s)}$ be the strategy tree for 
$\lambda_{1,\alpha}^\ell(s)$ with  initial state $s$.
We say that a node $v$ of $T_{\lambda_{1,\alpha}^{\ell}(s)}$ is an 
\emph{$\alpha$-good} node if the average of the weights in all dimensions 
in $I$ and dimension $\ell$ of the path from the root to $v$ is at least~$-\alpha$.
% and the length of the path is at least $i$.
The tree  $\widehat{T}_{\lambda_{1,\alpha}^{\ell,Z}(s)}$ is obtained from
$T_{\lambda_{1,\alpha}^{\ell}(s)}$ by removing all descendants of $\alpha$-good
nodes that are at depth at least $Z$. 
Hence, the leaves of $\widehat{T}_{\lambda_{1,\alpha}^{\ell,Z}(s)}$ are $\alpha$-good.

%\begin{proposition}\label{prop:FiniteTree_2}
%The tree $\widehat{T}_{\lambda_{1,\alpha}^{\ell,i}(s)}$ is a finite tree.
%\end{proposition}

%\begin{proof}
We show that $\widehat{T}_{\lambda_{1,\alpha}^{\ell,Z}(s)}$ is a finite tree.
By K\"{o}nig's Lemma~\cite{Konig36}, it suffices to show that every path in the tree is finite.
Assume towards contradiction that there is an infinite path $\pi$ 
in the tree. 
Hence $\pi$ is a play consistent with $\lambda_{1,\alpha}^{\ell}(s)$, and since 
$\pi$ does not contain any $\alpha$-good node, it follows
that for some dimension $i \in I \cup \{\ell\}$ we have $(\MPsup(\pi))_{i} \leq -\alpha$
(and $(\MPinf(\pi))_{i} \leq -\alpha$ as well).
It follows that $\pi \not\in \Phi_{\ell}(-\frac{\alpha}{2})$.
This contradicts the assumption that $\lambda_{1,\alpha}^{\ell}(s)$ is a winning strategy 
for player~1 for $\Phi_{\ell}(-\frac{\alpha}{2})$.
%\hfill\qed
%\end{proof}

We now describe a strategy for player~1 based on the finite-memory 
winning strategies for $\Phi_{\ell}(-\frac{\alpha}{2})$
and show that the strategy is winning for the objective $\MeanPayoffInfSup(I,J)$.

\begin{algorithmic}[1]
\STATE $\alpha \gets 1$ \LOOP \FOR{$\ell \in J$} \STATE Let $s$ be
the current state, and $L$ be the play length so far.
%%\mynote{L: check the bounds (see also Lemma~\ref{lem:sup}).}

\STATE $Z\gets \max\set{\frac{L\cdot W}{\alpha}, N^*_{\frac{\alpha}{2}}}$ 
(where $N^*_{\frac{\alpha}{2}}= \max\{ 
N_{\frac{\alpha}{2}, \wh{\lambda}(s), s} \mid s \in S, \ell' \in J, \wh{\lambda}(s)=\lambda_{1,\frac{\alpha}{2}}^{\ell'}(s)\}$,
that is, $\wh{\lambda}(s)=\lambda_{1,\frac{\alpha}{2}}^{\ell'}(s)$ is the finite-memory strategy for 
$\Phi_{\ell'}(-\frac{\alpha}{2})$ from $s$, 
the number $N_{\frac{\alpha}{2}, \wh{\lambda}(s),s}$ is as defined in 
Lemma~\ref{lemm:mixed_req_2} for the strategy,
and $N^*_{\frac{\alpha}{2}}$ is the maximum over $\ell' \in J$) 
%%\mynote{L: multi-subscript hard to read.}
%\LOOP
% \FOR{$m=1$ to $i$} 
\STATE Play according to $\lambda_{1,\alpha}^{\ell}(s)$ until a leaf $s'$ of 
$\widehat{T}_{\lambda_{1,\alpha}^{\ell,Z}(s)}$ is reached.
% \STATE  $s\gets s'$
% \ENDFOR
  % \ENDLOOP
  \ENDFOR \STATE $\alpha \gets
\frac{\alpha}{2}$ \ENDLOOP
\end{algorithmic}

Let $W \in \Nat$ be the largest absolute value of the weight function $w$.
After the last command in the internal for-loop
was executed, the mean-payoff value in dimension $\ell$, is at least
$\frac{- L\cdot W - Z \cdot \alpha}{L + Z}$ where $Z \geq \frac{L\cdot W}{\alpha}$
and this is at least 
\[
\frac{-L\cdot W - \alpha \cdot \frac{L\cdot W}{\alpha}}{L + \frac{L\cdot W}{\alpha}}
\geq -2\cdot \alpha.
\] %%for sufficiently large $L$. 
Consider the segment of the play for the round for a value of $\alpha$: let us denote
by $M_b$ the number of steps played till the beginning of the round and we will 
denote by $M_t$ the total number of steps of the current round.
Our goal is to obtain an upper bound on the average of the weights for all 
$n \leq M_b + M_t$.
In the beginning of the round (i.e., after $M_b$ steps) the average value for dimensions in $I$ 
is at least $-2\cdot \alpha$ (recall that $\alpha$ has been halved in line 8).   %  \mynote{L: why not $-\alpha$.}, 
Step~5 ensures that at least $N_{\alpha}^*$ 
steps have been already played, i.e.,  $M_b \geq N_{\alpha}^*$. 
It follows from Lemma~\ref{lemm:mixed_req_2} that for all dimensions in $I$ 
and for all steps $M_b\leq n \leq M_b +M_t$ of the current round, the sum of the weights 
is at least $-(M_b\cdot 2\cdot\alpha + N_{\alpha}^*\cdot\alpha + (n-M_b)\cdot \alpha),$ 
and hence the average value at step $n$ is at least 
\[
\frac{-(M_b \cdot 2 \cdot \alpha  + N_{\alpha}^* \cdot \alpha + (n-M_b)\cdot \alpha)}{n} \geq -4\cdot \alpha
\]
since $n \geq M_b$ and $n \geq N_{\alpha}^*$.
That is, for all steps in the round for $\alpha$, for all dimensions in $I$,
the average value is at least $-4\cdot \alpha$.
In every external for-loop $\alpha$ gets smaller, and $L$ gets bigger. 
Moreover, since the tree $\widehat{T}_{\lambda_{1,\alpha}^{\ell,Z}(s)}$ is finite, 
it follows that the main loop gets executed infinitely often (i.e., the strategy
does not get stuck in the for-loop). 
Thus when the length of the play tends to infinity, the supremum of the mean-payoff 
value tends to a value at least $0$ in every dimension $j \in J$, and 
the infimum of the mean-payoff value tends to a value at least~0
in every dimension $i \in I$.
Hence the strategy described above is a winning strategy for player~1.
\hfill\qed
\end{proof}

\begin{lemma}
%For all states $s \in S$, if there is a winning strategy for player~2 from $s$ to 
%falsify the objective $\MeanPayoffInfSup(I,J)$ against all player-1 strategies, 
%then there exists a memoryless winning strategy for player~2 from $s$ to falsify the 
%objective against all player-1 strategies.
In multi-mean-payoff games with  objective $\MeanPayoffInfSup(I,J)$
for player~$1$, memoryless strategies are sufficient for player~$2$.
\end{lemma}

\begin{proof}
%%[\ProofOfLem{lem:GMPSup}]
%\mynote{L: update this proof as in Lemma~\ref{lem:sup2}.}
The proof is similar to the proof of Lemma~\ref{lem:sup2}, and
based on induction on the number of states $\abs{S}$ in the game
structure.
The base case with $\abs{S} = 1$ is obvious.
We now consider the inductive case with $\abs{S} = n \geq 2$. 
%Let $k \in \Nat$ be the dimension of the weight function $w$, 
For $\ell \in J$, let $W_\ell$ be the winning region for player~2 for the 
objective $\Phi_{\ell}$ as defined in Lemma~\ref{lemm:mixed_1}.
Let $W=\bigcup_{\ell \in J} W_\ell$. We consider the following two cases:

\begin{enumerate}

\item 
If $W=\emptyset$, then player~1 can satisfy the objective $\Phi_{\ell}$ 
for all $\ell \in J$, and 
then by Lemma~\ref{lemm:mixed_1} player~1 wins from 
everywhere for the objective $\MeanPayoffInfSup(I,J)$.
Hence there is no winning strategy for player~2.

\item If $W\neq\emptyset$, then there exists $\ell \in J$ such that $W_\ell \neq \emptyset$.
In $W_\ell$ there is a memoryless winning strategy $\lambda_2$ for player~2 
to falsify $\Phi_{\ell}$, and the strategy also falsifies $\MeanPayoffInfSup(I,J)$
as $\MeanPayoffInfSup(I,J)=\bigcap_{\ell \in J} \Phi_{\ell}$.
The existence of memoryless winning strategy for player~$2$ follows from the following facts:
by Lemma~\ref{lem:GMPmixedOneSup} it follows that if player~2 can 
falsify the objective $\Phi_{\ell}$, then player~2 can also falsify the 
objective where in the dimension $\ell$ we consider the mean-payoff-inf 
objective instead of mean-payoff-sup objective, and 
the existence of memoryless strategies against mean-payoff-inf objectives
follows from Theorem~\ref{thrm:inf}.  
The rest of the proof is identical to the proof of Lemma~\ref{lem:sup2}
and can be omitted (we present it for sake of completeness). 
Since $W_\ell$ is a winning region for player~$2$ it follows that $W_\ell=\AttrTwo{W_\ell}$,
and hence the graph $G'$ induced by $S \setminus W_\ell$ is a game structure.
%Let $G'$ be the game structure induced by $S \setminus W_\ell$
Let $W' = W \setminus W_\ell$ be the winning region for player~2 in $G'$. 
By inductive hypothesis (since $G'$ has strictly fewer states as a 
non-empty set $W_\ell$ is removed), it follows that there is a memoryless
winning strategy $\lambda_2'$ in $G'$ for the region $W'$.
The winning region $S \setminus (W_\ell \cup W')$ for player~1 in $G'$
is also winning for player~1 in $G$ (since $W_{\ell} = \AttrTwo{W_{\ell}}$, 
$G'$ is obtained by removing only player~1 edges).
Hence to complete the proof it suffices to show that the memoryless
strategy obtained by combining $\lambda_2$ in $W_\ell$ and $\lambda_2'$ 
in $W'$ is winning for player~2 from $W_\ell \cup W'$.
Define the strategy $\lambda_2^*$ as follows:
\[ 
\lambda^*_2(s) = \left\{ 
	\begin{array}{ll}
         \lambda_2(s) & \qquad \mbox{if $s\in W_\ell$}\\
         \lambda_2'(s) & \qquad \mbox{if $s\in W'$}.
	\end{array} \right. 
\]
Consider the memoryless strategy $\lambda_2^*$ for player~2 and 
the outcome of any counter strategy for player~1 that starts in $W'\cup W_\ell$.
There are two cases:
(a)~if the play reaches $W_\ell$, then it reaches in finitely many steps,
and then $\lambda_2$ ensures that player~2 wins; and 
(b)~if the play never reaches $W_i$, then the play always stays in $G'$, 
and now the strategy $\lambda_2'$ ensures winning for player~2.
%%Hence the result follows.
\end{enumerate}
The desired result follows.
\hfill\qed
\end{proof}

\paragraph{\bf coNP upper bound.} 
Since memoryless winning strategies exist for player~2, 
to establish the coNP upper bound we need to show that 
one-player game structures with $\MeanPayoffInfSup(I,J)$ 
objectives can be solved in polynomial time. 
First we interpret $\MeanPayoffInfSup(I,J)$ as the 
conjunction of $\Phi_{\ell}$ for $\ell \in J$.
From Lemma~\ref{lem:GMPmixedOneSup} it follows
every $\Phi_{\ell}$ can be considered as $\MeanPayoffInf$ 
objective and hence can be solved in polynomial time for
one-player game structures by the results of 
Section~\ref{subsec:mean-inf}.
Hence the coNP upper bound follows.
We have the following theorem summarizing the results 
of this section.

\begin{theorem}\label{thrm:mixed}
For multi-weighted two-player game structures with objective 
$\MeanPayoffInfSup(I,J)= \{ \pi \in \Plays(G) \mid 
\forall i \in I: \MPinf(\pi)_i \geq 0 \text{ and } 
\forall j \in J: \MPsup(\pi)_j \geq 0 \}$
for player~1, the following assertions hold:
\begin{enumerate}
\item Winning strategies for player~1 require infinite-memory in general, 
and memoryless winning strategies exist for player~2.

\item The problem of deciding whether a given state is winning for player~1
is coNP-complete.

\end{enumerate}

\end{theorem}

%%\vspace{-1em}
\section{Conclusion}
In this work we considered games with multiple mean-payoff and energy 
objectives, and established determinacy under finite-memory, inter-reducibility of 
these two classes of games for finite-memory strategies, and improved the 
complexity bounds from EXPSPACE to coNP-complete. 
We also showed that multi-energy and multi-mean-payoff games under 
memoryless strategies are NP-complete.
Finally, we studied multi-mean-payoff games with infinite-memory strategies
and show that multi-mean-payoff games with mean-payoff-sup objectives 
can be decided in NP $\cap$ coNP (and can be solved in polynomial time
if mean-payoff games with single objective can be solved in polynomial time); 
and multi-mean-payoff games with mean-payoff-inf objectives, and combination of
mean-payoff-inf and mean-payoff-sup objectives are coNP-complete.
Thus we present optimal computational complexity results for multi-energy 
and multi-mean-payoff games under finite-memory, memoryless, and infinite-memory
strategies.

%\mynote{L: say some quick words about disjunction of mean-payoff objectives, + quick summary of results for disjunction.}

%%Two interesting problems are open: (A)~for multi-mean-payoff games, 
%%the winning strategies with infinite memory are more powerful than finite-memory strategies, 
%%and the complexity of solving multi-mean-payoff games with infinite-memory 
%%strategies remains open.
%%(B)~it is not known how to compute the exact or 
%%approximate Pareto curve (trade-off curve) for multi-objective mean-payoff and
%%energy games.
%%\section{Open questions}

\smallskip\noindent{\bf Acknowledgement.} 
We are grateful to Jean Cardinal for pointing the reference~\cite{KS88}.

%%\begin{itemize}
%%	\item on the complexity gap: should we explain the two possibilities that 
%%they exists to close the gap and what we conjecture ?
%%	\item it would be nice to have a bound on the initial credit needed by
%%Player~1 when she is winning in order to justify a "practical" fixed point algorithm...
%%\end{itemize}
%%\vspace{-1em}
%%{\small

%%}

\clearpage

\section*{Appendix}

We discuss the results of~\cite{Kop} which shows the existence of memoryless
winning strategies for player~2 when the objective for player~1 is the 
conjunction of mean-payoff-inf objectives. 
We will also argue that the results of~\cite{Kop} do not show the existence
of memoryless winning strategies for player~2 when the objective for player~1
is the conjunction of mean-payoff-sup objectives (the result that we 
establish in Lemma~\ref{lem:sup2}).
The result of~\cite{Kop} requires the notion of \emph{convexity} 
for \emph{prefix-independent} objectives.

\smallskip\noindent{\bf Prefix-independent and convex objectives.} 
An objective $\varphi$ is prefix-independent if for all plays 
$\pi$ and $\pi'$ such that $\pi'=\rho \cdot \pi$, where $\rho$ is a
finite prefix, we have $\pi \in \varphi$ iff $\pi'\in \varphi$, 
i.e., the objective is independent of finite prefixes.
A play $\pi$ is a \emph{combination} of two plays  $\pi_1=u_1 u_3 u_5 \ldots$ 
and $\pi_2=u_0 u_2 u_4 \ldots$, where $u_i$'s are finite prefixes, if 
$\pi= u_0 u_1 u_2 u_3 u_4 \ldots$.
An objective $\varphi$ is \emph{convex} if it is %%as a subset of $C^\omega$ it is 
closed under combination.
We refer the reader to~\cite{Kop} for further details.
The results of~\cite{Kop} shows that if the objective for player~1 is  
prefix-independent and convex, then memoryless winning strategies
exist for player~2.
It is easy to verify that mean-payoff-inf objectives are both 
prefix-independent and convex. 
It follows that conjunction of mean-payoff-inf objectives are also 
prefix-independent and convex.
Hence in games with conjunction of mean-payoff-inf objectives, memoryless
winning strategies exist for player~2.
We now show with an example that in contrast mean-payoff-sup objectives 
are not convex.

\begin{example}
Consider a one-player game structure $G$ with two states $\set{s_{+},s_{-}}$, 
with all edges, such that all incoming edges to state $s_{+}$ have weight 
$+2$, and all incoming edges to $s_{-}$ have weight $-2$. 

Consider the following play $\pi_0$:
\begin{enumerate}
\item \emph{Step~1.} Repeat the self-loop in $s_{-}$ until the average weight of 
the play prefix is below $-1$, then take edge to $s{+}$ and goto Step~2.
\item \emph{Step~2.} Repeat the self-loop in $s_{+}$ until the average weight of 
the play prefix is above $1$, then take edge to $s_{-}$ and goto Step~1.
\end{enumerate}
Consider the play $\pi_1$ obtained by exchanging $s_{+}$ and $s_{-}$ in $\pi_0$.
It is easy to verify that $\MPsup(\pi_0) = \MPsup(\pi_1) = +1$.
However, for the following combination of the plays $\pi_2$, such that forall 
$i \geq 0$ the $2i-1$-th state of $\pi_2$ is the $i$-th state of $\pi_0$ 
and the $2i$-th state of $\pi_2$ is the $i$-th state of $\pi_1$.
We get that $\MPsup(\pi_2) = 0$.
It follows that mean-payoff-sup objectives are not convex.
\end{example}

\end{document}